%% file: CLDS-arxiv.tex
\documentclass[10pt]{article}

\usepackage[a4paper]{geometry}

\usepackage{amsmath}
\usepackage{amssymb}
\usepackage{amsthm}
\usepackage{amscd}
\usepackage{stmaryrd}
\usepackage{enumerate}
\usepackage{fouriernc}
\usepackage{hyperref}
\usepackage{tikz-cd}
\usepackage{yfonts}

\theoremstyle{definition}
\newtheorem{Theorem}{Theorem}
\newtheorem{Algorithm}{Algorithm}
\newtheorem{Definition}[Theorem]{Definition}
\newtheorem{Remark}[Theorem]{Remark}
\newtheorem{Corollary}[Theorem]{Corollary}
\newtheorem{Proposition}[Theorem]{Proposition}

\newtheorem{Lemma}[Theorem]{Lemma}
\usepackage{mdframed}

\include{commands}

\newcommand{\sem}[1]{\left\llbracket #1 \right\rrbracket}

\title{On the Complexity of Robust Eventual Inequality Testing for C-Finite Functions}
\date{}

\author{
Eike Neumann\\ Swansea University\\ Swansea, Wales\\ \texttt{neumaef1@gmail.com}
}

\begin{document}
\maketitle

\begin{abstract}
    We study the computational complexity of a robust version of the problem of testing two univariate C-finite functions for eventual inequality at large times.
    Specifically, working in the bit-model of real computation, we consider the eventual inequality testing problem for real functions that are specified by homogeneous linear Cauchy problems with arbitrary real coefficients and initial values.
    In order to assign to this problem a well-defined computational complexity, we develop a natural notion of polynomial-time decidability of subsets of computable metric spaces which extends our recently introduced notion of maximal partial decidability.
    We show that eventual inequality of C-finite functions is polynomial-time decidable in this sense.
\end{abstract}

\input{Introduction}

\input{Proof-Outline}

\input{Coefficient-Functions}

\input{Non-Equality}

\section{Proof of Theorem \ref{Theorem: main theorem}}
\input{Inequality-Algorithm}

\begin{Lemma}
Algorithm \ref{Algorithm: Inequality} is correct and runs in polynomial time.
\end{Lemma}
\begin{proof}

\input{inequality-partial-correctness}
\input{inequality-termination}
\end{proof}

\bibliography{CLDS}
\nocite{*}

\newpage
\appendix

\input{Spectral-Distance}
\input{Vandermonde-Lemmas}
\input{Coefficient-Function-Theorem}
\input{Proposition-relating-coefficient-functions}

\end{document}

%% file: commands.tex
\renewcommand{\Re}{\operatorname{Re}}
\renewcommand{\Im}{\operatorname{Im}}

\newcommand{\R}{\mathbb{R}}
\newcommand{\C}{\mathbb{C}}
\newcommand{\Q}{\mathbb{Q}}
\newcommand{\N}{\mathbb{N}}
\newcommand{\Z}{\mathbb{Z}}

\newcommand{\norm}[1]{\left|\left|#1\right|\right|}

\newlength\arrowheight

\newcommand{\Set}[2]{\left\{#1 \mid #2\right\}}



%% file: Introduction.tex
\section{Introduction}

Linear dynamical systems capture the evolution of a wide range of real-world systems of interest in physics, biology, economics, and other natural sciences. In computing, they arise in connection with questions such as loop termination, model checking, and probabilistic or quantum computation.
Unsurprisingly, questions surrounding decidability of liveness and safety properties of such systems have a long history in theoretical computer science, dating back at least to the mid-1970s with the work of Berstel and Mignotte \cite{BerstelMignotte76} on decidable properties of the zero sets of linear recurrence sequences -- a line of research that itself goes back at least to the beginning of the 20th century, as exhibited by Skolem's famous structural result of 1934.
A simple-looking but notoriously difficult open problem in this area is the so-called \emph{Positivity Problem} \cite{OWSurvey15}: decide whether a given linear recurrence sequence or C-finite function is positive.
The Positivity Problem is known to be decidable for linear recurrences of order up to five, but known to likely fall outside the scope of current techniques for recurrences of order six or higher \cite{OWPositivityLowOrder14}. 
The situation is better for simple linear recurrences, where Positivity is known to be decidable up to order nine \cite{OWPositivitySimple14} and Ultimate Positivity -- to decide if a recurrence sequence is eventually positive -- is known to be decidable for all orders \cite{OWUltimateSimple14}.
Similar results have been obtained for the continuous-time analogues of these problems for C-finite functions \cite{BellEtAl10, ChonevLICS16, ChonevICALP16}.

Recently, there has been increased interest in the study of robust reachability questions for linear dynamical systems: 
in \cite{Rounding}, point-to-point reachability for linear systems with rounded orbits is studied.
In \cite{Pseudo}, a version of the Skolem Problem (to decide whether a recurrence sequence has a zero), where arbitrarily small perturbations are performed after each step of the iteration, is shown to be decidable.
Most notably, in \cite{Robust} two robust versions of the Positivity Problem are studied: in the first version, one is given a fixed rational linear recurrence equation and a rational ball of initial values and asked to decide whether the linear recurrence is positive for all initial values in the ball.
In the second version, one is given a linear recurrence equation and a single vector of initial values and asked if there exists a ball about the initial values such that all sequences satisfying the equation with initial values in the ball are positive.
The second version is shown to be decidable in polynomial space, while the first version is shown to be at least as hard as the standard Positivity Problem. 
Both results are shown to hold true also for analogous versions of the Skolem Problem.

We recently proposed \cite{DPLSR20} another natural robust variant of the Positivity Problem based on computable analysis.
Computable analysis \cite{WeihrauchBook, BrattkaHertling21} is the study of computation on data that is in general not known exactly but can be measured to arbitrary accuracy.
Algorithms are provided infinitary objects as inputs via streams of increasingly accurate approximations.
This allows one to provide arbitrary real numbers (as opposed to ``efficiently discrete'' countable subsets of real numbers such as rational or algebraic numbers) as inputs to algorithms.
Here, robustness is built into the definition of computation, for any decision that an algorithm makes can only depend on a finite approximation to its input.

We showed that the Positivity Problem and the Ultimate Positivity Problem for linear recurrences with arbitrary real parameters are maximally partially decidable (see below for a definition).
In this paper we improve one of these results by showing that the Ultimate Inequality Problem for C-finite functions which are specified by arbitrary real parameters is maximally partially decidable in polynomial time.
This is a significant improvement over our previous decidability result, which requires exponentially many queries to the theory of the reals.
The focus on the continuous-time setting is mainly for the sake of variety.
Our results apply, \textit{mutatis mutandis}, to linear recurrence sequences as well.
More concretely, we consider functions that are specified by homogeneous linear differential equations
\begin{equation}\label{eq: differential equation}
    f^{(n)}(t) + c_{n - 1}f^{(n - 1)}(t) + \dots + c_1 f^{(1)}(t) + c_0 f(t) = 0
\end{equation}
with constant coefficients $c_0,\dots, c_{n - 1} \in \R$ and initial values
$f^{(k)}(0) \in \R$, $k =0,\dots,n - 1$.
Functions of this type are also called \emph{C-finite}.
The Ultimate Inequality Problem\footnote{Over discrete inputs, the Ultimate Inequality Problem reduces to the Ultimate Positivity Problem, \textit{i.e.}, the special case of Ultimate Inequality where the second function is identically equal to zero.
However, the standard reduction maps robust instances of Ultimate Inequality to non-robust instances of Ultimate Positivity, which is why we consider Ultimate Inequality instead. 
For the same reason, we cannot assume without loss of generality that $n = m$ in the definition of the Ultimate Inequality Problem.}
is to decide, given $(c, u, d, v) \in \R^{2n} \times \R^{2m}$ if there exists $t_0 > 0$ such that the function $f$ satisfying the differential equation \eqref{eq: differential equation} with coefficients $c$ and initial values $f^{(k)}(0)$
is greater than or equal to 
the function $g$ satisfying an equation analogous to \eqref{eq: differential equation} with coefficient vector $d$ and initial values $g^{(k)}(0) =  v_k$
for all $t \geq t_0$.
Thus, we admit functions as inputs that are specified by Cauchy problems with arbitrary real coefficients and initial values.
Computation on such objects can be defined using computable metric spaces \cite{BrattkaPresser}.
A \emph{computable metric space} is a separable metric space $(X, d)$ together with a dense sequence $(x_k)_k$ and a computable map 
$\delta \colon \N^3 \to \Q$ such that 
$|\delta(k, \ell, n) - d(x_k, x_\ell)| < 2^{-n}$.
A point $x \in X$ can be provided as an input to an algorithm in the form of an infinite sequence $(k_n)_n$ of integers such that 
$|x - x_{k_n}| < 2^{-n}$. 
A sequence $(k_n)_n$ with this property is called a \emph{name} of $x$.
Note that we do not require names to be computable. 
In our computational model, the algorithm is given access to the name of its input as a black box.
As a consequence, algorithms may operate on all points of $X$, not just on the computable points.
This should be distinguished from the related notion of \emph{Markov computability} (see \cite[Chapter 9.6]{WeihrauchBook} and references therein), where algorithms operate on computable points which are presented as Gödel numbers of Turing machines.

The prototypical example of a computable metric space is the space $\R$ of real numbers with the usual Euclidean distance, where $(x_k)_k$ is a suitable enumeration of the rational numbers.
The computability of the map $\delta$ ensures that the distance function $d \colon X \times X \to \R$ is computable when $\R$ is given this structure of computable metric space.

Since an algorithm can only read a finite approximation to its input before committing to a decision, connected computable metric spaces do not have any non-trivial decidable subsets.
For this reason, decision problems have arguably been somewhat neglected by the computable analysis community in the past. 
In order to extend the definition of decidability from $\N$ to arbitrary computable metric spaces in a more meaningful way, we proposed \cite{DPLSR20} the following notion:

    Let $A \subseteq X$ be a subset of a computable metric space $X$.
    A \emph{partial algorithm} for deciding $A$ is an algorithm that takes as input (a name of) a point\footnote{It is an essential requirement that the behaviour of a maximal partial algorithm be constrained on the entire set $X$, and not just on $X \setminus \partial A$. For example, it can be shown that for a real number $x \in \R$, the set $\R \setminus \{x\}$ is maximally partially decidable if and only if $x$ is computable.
    If the algorithm's behaviour was only constrained on $X \setminus \partial A$, then the algorithm which always outputs ``false'' would ``maximally partially decide'' all sets $A$ with $A \subseteq \partial A$.} 
    $x \in X$
    and either halts and outputs a boolean value or runs indefinitely.
    Should the algorithm halt on an input $x$, it is required to correctly report whether $x$ belongs to $A$.
    A \emph{maximal partial algorithm} for deciding $A$ is a partial algorithm for deciding $A$ which halts on all names of all points 
    $x \in X \setminus \partial A$.
    If there exists a maximal partial algorithm for deciding $A$, then $A$ is called \emph{maximally partially decidable}.
    A point in $X \setminus \partial A$ is also called a \emph{robust instance} of $A$, while a point in $\partial A$ is called a \emph{boundary instance} of $A$.
    A maximal partial algorithm for deciding a set $A$ can be alternatively defined as a partial algorithm whose halting set\footnote{We did not require in our definition of partial algorithm that the halting behaviour of a partial algorithm be independent of the given name of a point. The halting set is hence a-priori a set of names, and not a set of points in $X$. However, the map that sends names to the points they encode is ``effectively open'', so that every partial algorithm can be made into a partial algorithm with extensional halting behaviour.} 
    is maximal among the halting sets of all partial algorithms for deciding $A$.
    This motivates the name ``maximal partial algorithm''.

For metric spaces $X$ whose closed balls are compact, maximal partial decidability can be further characterised using rational balls as inputs (cf.~\cite[Proposition 2.2]{DPLSR20}).
This helps clarify how maximal partial decidability compares with notions of robust decidability proposed by other authors.
Let $(x_k)_k$ denote the dense sequence that is part of the data by which $X$ is presented.
The existence of a maximal partial algorithm for deciding $A$ is (uniformly computably) equivalent to the existence of an algorithm which takes as input an integer $k \in \N$ and a positive rational number $r \in \Q$ and behaves as follows:
\begin{enumerate}
    \item If the closed ball $\overline{B}(x_k,r)$ is contained in the set of robust problem instances, then the algorithm halts and correctly reports the uniform answer to the decision problem on $\overline{B}(x_k,r)$.
    \item If the open ball $B(x_k, r)$ intersects the set of boundary instances, then the algorithm halts and correctly reports that there exist both ``Yes''-instances and ``No''-instances in the ball.
    \item If the closed ball intersects the boundary instances but the open ball does not, then the algorithm runs forever.
\end{enumerate}

In this paper, we extend the definition of maximal partial decidability to polynomial-time decidability.
For general background on computational complexity in analysis, see \cite{KoBook}.
Using second-order polynomials \cite{KawamuraCook} 
and parametrised spaces \cite{NeumannSteinberg18},
the notion of ``maximal partial decidability in polynomial time'' can be defined in great generality.
For our present purpose, an ad-hoc definition will suffice.
We restrict ourselves to computable metric spaces $X$ such that there exists a function 
$\operatorname{size} \colon X \to \N$
with the property that every point $x \in X$ is the limit of a sequence $(x_{k_n})_n$ with $|x_{k_n} - x| < 2^{-n}$ and
$k_n \leq \operatorname{size}(x) + O(n)$.
This property is closely related to $\sigma$-compactness of $X$.
See \cite{Schroder04, KunkleSchroeder05} for discussions on spaces that admit size-functions of this type.
We tacitly assume in the sequel that all inputs to our algorithms are guaranteed to satisfy a size bound of this form.

The only computable metric spaces we consider in this paper are subspaces of countable sums of finite products of $\R$ or $\C$.
Size functions for these spaces can easily be constructed from size functions for the spaces $\mathbb{K}^n$ where $\mathbb{K} \in \{\R, \C\}$.
For the latter, we may put
\[
    \operatorname{size}\left(\left(x_1,\dots,x_{n}\right)\right) = O\left(n + \sum_{j = 1}^{n} \left\lceil\log_2(|x_j| + 1)\right\rceil\right).
\]

    Let $X$ be a computable metric space admitting a size function as above.
    Let $A \subseteq X$ be a subset of $X$.    
    We say that $A$ is maximally partially decidable in polynomial time if there exist
    a polynomial $P \in \N[x]$
    and a maximal partial algorithm for deciding $A$ 
    such that given a name of a point $x \in X \setminus \partial A$,
    the algorithm halts within at most 
    \[P\left(\operatorname{size}(x) - \log \left(\min\left\{d(x,\partial A), 1\right\}\right)\right)\]
    steps
    (with the convention that $\min\{d(x,\partial A), 1\} = 1$ if $\partial A = \emptyset$).
Observe that we recover the definition of polynomial-time decidability of subsets of $\N$ by interpreting $\N$ as a computable metric space with the discrete metric.

For 
$\mathbb{K} \in \left\{\R, \C\right\}$
and $n \geq 0$,
write $C_n(\mathbb{K}) = \mathbb{K}^n \times \mathbb{K}^n$.
Let 
$\sem{\cdot} \colon \sum_{n \in \N} C_n(\mathbb{K}) \to C(\R,\mathbb{K})$
be the function that maps $(c, u) \in C_n(\mathbb{K})$ to the unique solution $f \colon \R \to \mathbb{K}$ 
of the Cauchy problem \eqref{eq: differential equation} 
with coefficient vector $c$ and initial values $u$.
By convention, if $n = 0$, then the single element of $C_n(\mathbb{K})$ represents the constant zero function.
The following is our main result:

\begin{Theorem}\label{Theorem: main theorem}
    The Ultimate Inequality Problem is maximally partially decidable in polynomial time.
    More precisely, given $(c, u, d, v) \in C_n(\R) \times C_m(\R)$ we can maximally partially decide in polynomial time if 
    $\sem{(c,u)}(t) \geq \sem{(d,v)}(t)$
    for all sufficiently large $t$.
\end{Theorem}

As part of the proof of Theorem \ref{Theorem: main theorem}, we establish the following result of independent interest:

\begin{Theorem}\label{Theorem: equality checking}
    Equality comparison is  maximally partially decidable in polynomial time.
    More precisely, given $(c, u, d, v) \in C_n(\C) \times C_m(\C)$ we can maximally partially decide in polynomial time if 
    $\sem{(c,u)} \neq \sem{(d,v)}$.
\end{Theorem}

It can be proved essentially as in \cite[Proposition 10.1]{DPLSR20} that the sets of boundary instances of both problems above have measure zero,
so that our algorithms halt on almost every input.

%% file: Proof-Outline.tex
\section{Proof Outline}

Let us briefly outline the proof of Theorem \ref{Theorem: main theorem}.
For the purpose of exposition, consider the special case of the Ultimate Positivity Problem, \textit{i.e.}, the special case of Ultimate Inequality where the second function is identically equal to zero\footnote{Formally speaking, this is the special case where the second input is an element of $C_0(\R)$, so that we know for certain that the input represents the constant zero function.}.

The polynomial
\[
    \chi_{c}(z) = z^n + c_{n - 1}z^{n - 1} + \dots + c_1 z + c_0
\]
is called the characteristic polynomial of (the differential equation with coefficient vector) $c$.
The roots of $\chi_c$ are also called the \emph{characteristic roots} of (the differential equation with coefficient vector) $c$.
Write $\sigma_c \subseteq \C$ for the set of all roots of $\chi_c$.
The Cauchy problem \eqref{eq: differential equation} has a unique solution $\sem{(c,u)} \colon \R \to \C$, which has the shape 
\begin{equation}\label{eq: exponential polynomial with indexes}
    \sem{(c,u)}(t) = 
    \sum_{\substack{\lambda \in \sigma_c}} P_{\lambda}(c, u, t) e^{\lambda t} =
    \sum_{j = 1}^{N} \sum_{k = 0}^{m_j - 1} a_{j, k} t^k e^{\lambda_j t}.
\end{equation}
The representation \eqref{eq: exponential polynomial with indexes} is also called the \emph{exponential polynomial solution} of \eqref{eq: differential equation}.

A result by Bell and Gerhold \cite{BellGerhold07} asserts that non-zero linear recurrence sequences without positive real characteristic roots admit positive and negative values infinitely often. 
As an immediate consequence, one obtains an analogous result for linear differential equations (see also the proof of \cite[Theorem 12]{BellEtAl10}):

\begin{Corollary}\label{Corollary: C-finite function without real characteristic roots}
    Let $f$ satisfy a linear differential equation \eqref{eq: differential equation} where $\chi_c$ has no real roots.
    Then either $f$ is identically zero, or there exist unbounded sequences 
    $(t_j)_j$ and $(s_j)_j$ of positive real numbers with 
    $f(t_j) < 0$
    and 
    $f(s_j) > 0$.
\end{Corollary}


It is well-known \cite[p. 117]{KoBook} that the roots of a polynomial can be computed in polynomial time in the sense described above\footnote{This can be deduced from the many results on efficient root finding for polynomials with rational coefficients, starting with the famous result of Schönhage \cite{Schoenhage}.}:
There exists a polynomial-time algorithm which takes as input a vector 
$\left(c_0,\dots,c_{n - 1}\right) \in \C^n$ 
of complex numbers and returns as output a vector 
$\left(z_0,\dots,z_{n - 1}\right) \in \C^n$
of complex numbers such that $(z_0,\dots,z_{n - 1})$ contains all roots of the polynomial 
$\chi_c(z)$,
counted with multiplicity.
Any such algorithm is necessarily multi-valued, which means that it is allowed to return different outputs
$(z_0,\dots,z_{n - 1})$
for different names of the same input 
$(c_0,\dots,c_{n - 1})$.
Out course, any two valid outputs for the same input agree up to permutation. 

It is relatively easy to see that $\sem{(c, u)}$ is robustly eventually positive if and only if $\chi_c$ has a simple real root $\rho$ with 
$\rho > \Re(\lambda)$ for all other roots such that the coefficient of $\rho$ in the exponential polynomial solution \eqref{eq: exponential polynomial with indexes} is strictly positive. 
We will show that we can detect this situation and in this case compute the sign of the coefficient in polynomial time.

The case where $\sem{(c, u)}$ robustly fails to be eventually positive is much more difficult.
Intuitively, this can happen for two (potentially overlapping) reasons:
    \begin{enumerate}[A]
        \item\label{Reason A} The polynomial $\chi_c$ has a root $\lambda$ with non-zero imaginary part such that 
            $\Re(\lambda) > \rho$ for all real roots $\rho$ of $\chi_c$
            and 
            the coefficient $P_{\lambda}(c, u, t)$ in \eqref{eq: exponential polynomial with indexes} is non-zero.
            In this case, eventual inequality fails due to Corollary \ref{Corollary: C-finite function without real characteristic roots}.
        \item\label{Reason B}
            The polynomial $\chi_c$ has real roots
            and 
            the leading coefficient of the polynomial $P_{\rho}(c, u, t)$ in \eqref{eq: exponential polynomial with indexes} corresponding to the largest real root $\rho$
            is strictly negative, robustly under small perturbations of $(c, u)$.
    \end{enumerate}
    
    There is no obvious reason why either of these properties should be semi-decidable, since the polynomials
    $P_{\lambda}(c, u, t)$
    do not depend continuously on $c$.
    Indeed, the degree of $P_{\lambda}(c, u, t)$ depends on the multiplicity of the roots of $\chi_c$, which 
    is unstable under small perturbations.
    When the multiplicities of the roots are fixed, the polynomials $P_{\lambda}$ depend computably on the initial data.
    The function which sends the initial data to the coefficients of the $P_{\lambda}$'s subject to fixed multiplicities is unbounded near the boundary of its domain.
    This implies that if $\lambda$ has multiplicity $> 2$, then the coefficient of $P_{\lambda}$ will ``jump discontinuously'' from a finite value to an arbitrarily large value under arbitrarily small perturbations of the initial data.
    
    In order to verify Condition \eqref{Reason A} we approximate the roots of $\chi_c$ to finite error $2^{-N}$.
    We can then identify roots $\lambda_1,\dots,\lambda_m$ that are guaranteed to have non-zero imaginary part and real part larger than any real root.
    To verify that one of the coefficients $P_{\lambda_j}(c, u, t)$ does not vanish, we construct the differential equation $c'$ with 
    $\chi_{c'}(z) = \chi_c / \prod_{j=1}^m(z - \lambda_j)$
    and check that the functions $\sem{(c,u)}$ and $\sem{(c', u)}$ are different (to some finite accuracy).
    In case this check does not succeed, we start over with increased accuracy.

In order to verify Condition \eqref{Reason B}, observe that we can compute arbitrarily good upper bounds on the multiplicities of the roots.
If we fix multiplicities $m_1,\dots,m_N$ with $m_1 + \dots + m_N = n$, then the coefficients $P_{\lambda}(c, u, t)$ depend continuously on $c$ and $u$.
The function 
$F_{m_1,\dots,m_N}(\lambda_1,\dots,\lambda_N,u_0,\dots,u_{n - 1})$ 
which sends a vector of roots $(\lambda_1,\dots,\lambda_N)$ and initial values $(u_0,\dots,u_{n - 1})$ to the leading coefficient
of the polynomial $P_{\lambda_1}(c, u, t)$, where $\chi_c = \prod_{j = 1}^N(z - \lambda_j)^{m_j}$,
is rational.
It has singularities at all points where $\lambda_j = \lambda_1$ for $j > 1$.
We show in Theorem \ref{Theorem: coefficient functions} below that 
$
F_{m_1,\dots,m_N}(\lambda_1,\dots,\lambda_N,u_0,\dots,u_{n - 1})
=
F_{m_1,1,\dots,1}(\lambda_1, \lambda_2,\dots,\lambda_2,\dots,\lambda_N,\dots,\lambda_N,u_0,\dots,u_{n - 1})$
and 
\[
    F_{m_1,1,\dots,1}(\lambda_1,\lambda_{m_1 + 1},\dots,\lambda_n,u_0,\dots,u_{n - 1}) = 
        \frac{G_{m_1,n}(\lambda_1,\lambda_{m_1 + 1},\dots,\lambda_n,u_0,\dots,u_{n - 1})}
             {(m_1 - 1)! \prod_{j = m_1 + 1}^n (\lambda_j - \lambda_1)}
\]
where 
$
G_{m_1,n}(\lambda_1,\lambda_{m_1 + 1},\dots,\lambda_n,u_0,\dots,u_{n - 1})
$
depends linearly on $u_0,\dots,u_{n - 1}$, the coefficient of $u_{n - 1}$ being $(-1)^{n + m_1}$. 
These results are obtained by studying the cofactor expansions of generalised Vandermonde matrices.
Now, to verify Condition \eqref{Reason B}, we first approximate the characteristic roots to error $2^{-N}$ for some integer $N$.
We identify the largest ``potentially real'' root $\rho$, \textit{i.e.} the root with largest real part among those roots whose imaginary part has absolute value less than $2^{-N}$.
We compute an upper bound on the multiplicity $m_1$ of this root.
We then evaluate the ``potential numerators'' 
$\Re G_{\ell,n}(\rho,\rho,\dots,\rho,\lambda_{m_1 + 1},\dots,\lambda_n, u)$ 
for $\ell = 1,\dots,m_1$
as far as the accuracy of the root approximations allows,
to check if the sign of the real part of the leading coefficient of $P_{\rho}(c, u, t)$ is guaranteed to be negative. 
It is possible for this check to succeed while $\rho$ is not a real root.
However, in that case $\Re \rho$ is larger than all real roots and $P_{\rho}(c, u, t)$ is non-zero, so that Condition \eqref{Reason A} is met.
If the check does not succeed, we start over, increasing the accuracy to $N + 1$.

In the full algorithm, we run all three of the above searches simultaneously while increasing the accuracy to which we approximate the roots.
To establish polynomial running time, we show that if our searches do not succeed after ``many'' steps, then there exists a ``small'' perturbation of the problem which is eventually positive, and another ``small'' perturbation which is not eventually positive.
This is made possible mainly due to the special shape of the polynomials
$
G_{m_1,n}
$.
If the search to verify Condition \eqref{Reason A} does not halt after ``many'' steps, the coefficients $P_{\lambda}(c, u, t)$ of all roots with ``large'' imaginary part whose real part is larger than the real part of the largest root with ``small'' imaginary part can be made zero by a ``small'' perturbation\footnote{Observe that the coefficients $P_{\lambda}(c, u, t)$ themselves are not necessarily small, since large coefficients for close characteristic roots could cancel each other out. In this case, a small perturbation of $(c, u)$ can induce a large perturbation of $P_{\lambda}(c, u, t)$. It is not difficult to construct instances where this is the case.} of $(c, u)$. 
The argument to establish this is elementary, but less obvious than one might expect (at least in the general case of Ultimate Inequality).
The same argument yields a proof of Theorem \ref{Theorem: equality checking}.


In perturbation arguments we will almost always perturb the roots of the characteristic polynomial $\chi_c$, rather than perturbing its coefficients directly.
It will therefore be useful to work with a suitable distance function.
Let $(c, u), (c', u') \in \mathbb{K}^{2n}$.
Let $\lambda_1,\dots,\lambda_n \in \C$ denote the roots of $\chi_{c}$, listed with multiplicity.
Let $\mu_1,\dots,\mu_n \in \C$ denote the roots of $\chi_{c'}$, listed with multiplicity.
Define the \emph{spectral distance} $d_{\sigma}(c, c')$ of $c \in \mathbb{K}^n$ and $c' \in \mathbb{K}^n$ as
$
    d_{\sigma}\left(c, c'\right) = \inf_{\pi \in S_n} \max_{j = 1,\dots,n}\left\{\left| \lambda_j - \mu_{\pi(j)} \right|\right\}
$.
The spectral distance $d_{\sigma}((c, u), (c', u'))$ of $(c, u) \in C_n(\mathbb{K})$ and $(c', u') \in C_n(\mathbb{K})$ is then defined as
\[
    d_{\sigma}\left((c, u), (c', u')\right) = \norm{u - u'}_{\infty} + d_{\sigma}\left(c, c'\right).
\]

The next proposition, proved in Appendix \ref{Appendix: Spectral Distance}, shows that the running time of an algorithm is bounded polynomially in the spectral distance to the set of boundary instances, then it is bounded polynomially in the ordinary distance. 

\begin{Proposition}\label{Proposition: spectral distance}
    Let $c, c' \in \mathbb{K}^{n}$ with $c \neq c'$.
    Then we have:
    \[
        - \log d_{\sigma}(c,c')
        \leq 
        - \log d(c,c')  + 2n \log (n) +  2n \log (\max\{\norm{c}_{\infty}, \norm{c'}_\infty\}).
    \]
\end{Proposition}

%% file: Coefficient-Functions.tex
\section{Leading Coefficient Functions}

We now study the leading coefficients of the polynomials $P_{\lambda}(c, u, t)$ in \eqref{eq: exponential polynomial with indexes} as functions of $c$ and $u$.

Let us introduce some useful notation.
We write $[a_q]^{q = 1,\dots,n}$ for the row vector $(a_1,\dots,a_n)$.
We write $[b_p]_{p = 1,\dots,m}$ for the column vector $(b_1,\dots,b_m)^T$.
We write 
$[c_{p,q}]_{p = 1,\dots, m}^{q = 1,\dots, n}$ 
for the $(m \times n)$-matrix with columns $[c_{p,q}]_{p = 1,\dots, m}$, where $q = 1,\dots,n$.
We also allow nesting of this notation. 
For example,
$\left[
    \left[a_{p, q, j}\right]_{p = 1,\dots, m}^{q = 1,\dots,n}
\right]^{j = 1,\dots,\ell}$
denotes the $(m \times (\ell \cdot n))$-matrix whose columns are 
$
    \left[a_{p, 1, 1}\right]_{p = 1,\dots, m}
\dots,
\left[a_{p, m, 1}\right]_{p = 1,\dots, m},
\dots,
\left[a_{p, 1, \ell}\right]_{p = 1,\dots, m},
\dots,
\left[a_{p, m, \ell}\right]_{p = 1,\dots, m}
$.

We write
$\left[[a_{p, q}]_{p = 1,\dots, m}^{q = 1,\dots,n}; [b_{p, q}]_{p = 1,\dots, m}^{q = n + 1, \dots, \ell}\right]$
for the $(m \times \ell)$-matrix whose first $n$ columns are the columns of $[a_{p, q}]_{p = 1,\dots, m}^{q = 1,\dots,n}$
and whose last $\ell - n$ columns are the columns of $[b_{p, q}]_{p = 1,\dots, m}^{q = n + 1, \dots, \ell}$, 
and so on.



In the sequel it will be convenient to extend the definition of the binomial coefficient $\binom{n}{k}$  to all integers $n, k \in \Z$, by letting $\binom{n}{k} = 0$ if $n$ or $k$ is negative.
For $k \in \Z$, we let $X^{[k]} = X^k$ if $k \geq 0$ and $X^{[k]} = 0$ if $k < 0$.

Let $n \geq 1$ and let $m_1,\dots,m_N$ be positive integers with $m_1 + \dots + m_N = n$.
The \emph{generalised Vandermonde matrix} with signature $(m_1,\dots,m_N)$ is the $(n\times n)$-matrix 
\[
    V_{m_1,\dots,m_N}(\Lambda_1,\dots,\Lambda_N) =
    \left[
        \left[\binom{p - 1}{q - 1}\Lambda_j^{[p - q]}\right]_{p = 1,\dots,n}^{q = 1,\dots,m_1}
    \right]^{j = 1,\dots,N}.
\]
The \emph{modified generalised Vandermonde matrix} with signature $(m_1,\dots,m_N)$ is the $(n\times n)$-matrix 
\[
    \widetilde{V}_{m_1,\dots,m_N}(\Lambda_1,\dots,\Lambda_N) =
    \left[
        \left[(q - 1)! \binom{p - 1}{q - 1}\Lambda_j^{[p - q]}\right]_{p = 1,\dots,n}^{q = 1,\dots,m_1}
    \right]^{j = 1,\dots,N}.
\]

It is easy to see that the coefficients 
$\left[
    \left[a_{1, k}\right]_{k = 0,\dots,m_1 - 1};
    \dots;
    \left[a_{N, k}\right]_{k = 0,\dots,m_N - 1}
\right]$ 
in \eqref{eq: exponential polynomial with indexes} satisfy the linear equation 
\begin{equation}\label{eq: Vandermonde equation}
    \widetilde{V}_{m_1,\dots,m_N}(\lambda_1,\dots,\lambda_N) \cdot \left[ \left[a_{j, k}\right]_{k = 0,\dots,m_1 - 1}\right]_{j = 1,\dots,N} = [u_j]_{j = 0,\dots,n - 1}.
\end{equation}

Equation \eqref{eq: Vandermonde equation} can be solved explicitly by computing the inverse of $\widetilde{V}_{m_1,\dots,m_N}(\lambda_1,\dots,\lambda_N)$
via cofactor expansion.
This motivates the following definition:
\begin{Definition}
    Let $n$ be a positive integer.
    Let $m_1,\dots,m_N$ be positive integers with $m_1 + \dots + m_N = n$.
    The \emph{leading coefficient function with signature $(m_1,\dots,m_N)$} is the function
    \begin{align*}
        &F_{m_1,\dots,m_N} \colon D_{m_1,\dots,m_N} \times \R^n \to \C,\\
        &F_{m_1,\dots,m_N}\left(\lambda_1,\dots,\lambda_N\right)\left(u_0,\dots,u_{n - 1}\right)
        = 
        \frac{\sum_{j = 1}^n (-1)^{j + m_1} \det\left(\widetilde{V}_{m_1,\dots,m_N}(\lambda_1,\dots,\lambda_N)_{j, m_1}\right) u_{j - 1} }
             {\det \widetilde{V}_{m_1,\dots,m_N}\left(\lambda_1,\dots,\lambda_N\right)},
    \end{align*}
    where 
    $
        D_{m_1,\dots,m_N} = \Set{(\lambda_1,\dots,\lambda_N) \in \C^N}
                {\lambda_j \neq \lambda_k \text{ whenever } j \neq k}.
    $
    Here, for a matrix $M$, the matrix $M_{j,k}$ is obtained from $M$ by deleting the $j^{\text{th}}$ row and $k^{\text{th}}$ column.
\end{Definition}
By our previous observation, the function $F_{m_1,\dots,m_N}$ 
maps the vector $(\lambda_1,\dots,\lambda_N, u_0,\dots,u_{n - 1})$ 
to the coefficient $a_{1, m_1 - 1}$ in \eqref{eq: exponential polynomial with indexes}.
We obtain the following result on the shape of $F_{m_1,\dots,m_N}$:

\begin{Theorem}\label{Theorem: coefficient functions}\hfill 
    \begin{enumerate}
        \item 
        Let $n \geq 1$. Let $1 \leq m_1 \leq n$.
        We have 
        \[
            F_{m_1,1,\dots,1}(\Lambda_1, \Lambda_{m_1 + 1}, \dots \Lambda_n)(U_0, \dots, U_{n - 1})
            =
            \frac{G_{m_1, n}\left(\Lambda_1,\Lambda_{m_1 + 1},\dots,\Lambda_n, U_0,\dots,U_{n - 1}\right)}{(m_1 - 1)! \prod_{j = m_1 + 1}^n
            \left(\Lambda_j - \Lambda_1\right)},
        \]
        where 
        \[
            G_{m_1, n}\left(\lambda_1,\lambda_{m_1 + 1},\dots,\lambda_n, u_0,\dots,u_{n - 1}\right)
            =
                \sum_{j = 1}^{n} (-1)^{j + m_1} A_{n,m_1,j}(\lambda_1,\lambda_{m_1 + 1}, \dots \lambda_n) u_{j - 1}
        \]
        with integer polynomials $A_{n,m_1,j} \in \Z\left[\Lambda_1,\Lambda_{m_1 + 1}, \dots \Lambda_n\right]$
        satisfying $A_{n, m_1, n} = 1$.
        Moreover, the polynomials $A_{n,m_1,j}$ can be evaluated on a complex vector in polynomial time, uniformly in 
        $n$, $m_1$, and $j$ (when these integers are given in unary).
        \item 
            There exists a polynomial $\Omega \in \N[X,Y,Z]$ with the following property:
            Let

            $\left(\lambda_1,\lambda_{m_1 + 1},\dots,\lambda_n, u_0,\dots,u_{n - 1}\right)$
            and 
            $\left(\mu_1,\mu_{m_1 + 1},\dots,\mu_n, v_0,\dots,v_{n - 1}\right)$
            be complex vectors.
            Let $C$ be a positive integer with
            $
                C \geq \log \left(\max\{\norm{u}_{\infty}, \norm{v}_{\infty}, \norm{\lambda}_{\infty}, \norm{\mu}_{\infty}\} + 1\right).
            $
            
            If 
            $|\lambda_j - \mu_j| < 2^{-\Omega(n, C,p)}$
            and 
            $|u_j - v_j| < 2^{-\Omega(n, C,p)}$
            for all $j$,
            then 

            $
              \left|
                G_{m_1, n}\left(\lambda_1,\lambda_{m_1 + 1},\dots,\lambda_n, u_0,\dots,u_{n - 1}\right)
                -
                G_{m_1, n}\left(\mu_1,\mu_{m_1 + 1},\dots,\mu_n, v_0,\dots,v_{n - 1}\right)
              \right|   
              <
              2^{-p}.
            $
        \item 
            Let $n \geq 1$. Let $m_1,\dots,m_N$ be positive integers with $m_1 + \dots + m_N = n$.
            Then 
            \begin{align*}
                F_{m_1,m_2,\dots,m_N}&(\Lambda_1,\dots,\Lambda_N)(U_0,\dots,U_{n-1})\\
                =&F_{m_1,\underbrace{1,\dots,1}_{m_2 \text{ times}},\dots, \underbrace{1,\dots,1}_{m_N \text{ times}}}
                (\Lambda_1, \underbrace{\Lambda_2,\dots,\Lambda_2}_{m_2 \text{ times}},\dots,\underbrace{\Lambda_N, \dots, \Lambda_N}_{m_N \text{ times}})(U_0,\dots,U_{n-1}).
            \end{align*}
\end{enumerate}
\end{Theorem}

To prove Theorem \ref{Theorem: coefficient functions}, we will compute the minors of the matrix $V_{m_1, 1, \dots, 1}$ by elementary row and column operations.
The minors of $\widetilde{V}_{m_1, 1, \dots, 1}$ differ from these only by a product of factorials.
Some of the results below may well be ``folklore'', but I am not aware of a reference.

During the calculation of the minors through elementary row and column operations one encounters a particular family of polynomials with integer coefficients.
For $m, n, i \in \Z$, let
\[
    Q_{n, m, i}(X_1,\dots, X_n) = \sum_{k_1 + \dots + k_n = m} \binom{k_1 + i}{i} X_n^{[k_n]}\cdot\dots\cdot X_1^{[k_1]}.
\]
Observe that a generalised Vandermonde block of dimension $n$ and multiplicity $m$ in the variable $\Lambda$ can be written as
$
    \left[
        Q_{1, p - q, q - 1}(\Lambda) 
    \right]_{p = 1,\dots,n}^{q = 1,\dots,m}.
$
The polynomials $Q_{n, m, i}$ satisfy recursive equations that allow us to calculate the minors.
A proof is given in Appendix \ref{Appendix: Vandermonde lemmas}.

\begin{Lemma}\label{Lemma: recursive equations for Q}
    The polynomials $Q_{n, m, i}$ satisfy the following recursive equations:
    \begin{enumerate}
        \item 
        $
        Q_{n, m, 0}(X_1,\dots, X_{n - 1}, Y) - Q_{n, m, 0}(X_1,\dots,X_{n - 1}, Z)
            =
        (Y - Z) Q_{n + 1, m - 1, 0}(X_1,\dots,X_{n - 1}, Y, Z)
        $
        for all $n, m \geq 1$.
        \item 
        $
        Q_{n + 1, m, i}(X_1,\dots,X_{n + 1}) - Q_{n, m, i + 1}(X_1,\dots,X_n)
        =
        (X_{n + 1} - X_1) Q_{n + 1, m - 1, i + 1}(X_1,\dots,X_{n + 1})
        $
        for all $i \geq 0$ and $n, m \geq 1$.
    \end{enumerate}
\end{Lemma}

Using these equations, we obtain a concrete representations of the minors, which allows us to prove Theorem \ref{Theorem: coefficient functions}.

\begin{proof}[Proof of Theorem \ref{Theorem: coefficient functions}]
    \hfill
    \begin{enumerate}
        \item
            Using Lemma \ref{Lemma: recursive equations for Q} one establishes a concrete representation of the minors of 
            $\widetilde{V}_{m_1,1,\dots,1}$.
            This is done in Lemma \ref{Lemma: minors of generalised Vandermonde matrix} in Appendix \ref{Appendix: Vandermonde lemmas 2}.
            The rest is a straightforward calculation.
            A full proof is given in Appendix \ref{Appendix: coefficient function theorem}.
        \item 
            This follows from the previous item since a polynomial-time computable function can only query its input to a polynomially bounded accuracy before committing to an output of a certain accuracy (cf.~\cite[Corollary 2.20]{KoBook}). 
            The full proof is given in Appendix \ref{Appendix: coefficient function theorem}.
        \item   
        Our explicit representation of the function $F_{m_1,1,\dots,1}$ shows that it extends to the complement of the set 
        \[
            Z' = \Set{(\lambda_1, \lambda_{m_1 + 1}, \dots, \lambda_n, u_0,\dots,u_{n - 1}) \in \C^{n - m_1 + 1} \times \C^{n}}{\exists j \geq m_1 + 1. \lambda_j = \lambda_1}.
        \]
        Let 
        \[\gamma \colon \C^{n - m_1 + 1} \times \C^n \setminus Z' \to \C\]
        be the function which sends a given complex vector 
        \[
            (\lambda_1, \lambda_{m_1 + 1},\dots,\lambda_{n},u_0,\dots,u_{n - 1})
        \]
        to the coefficient of $t^{m_1 - 1} e^{t \lambda_1}$ in the exponential polynomial solution \eqref{eq: exponential polynomial with indexes}
        to the initial value problem with characteristic polynomial 
        $\chi(z) = (z - \lambda_1)^{m_1} \prod_{j = m_1 + 1}^n (z - \lambda_n)$ 
        and initial values 
        $(u_0,\dots,u_{n - 1})$.
        This function agrees with $F_{m_1, 1, \dots, 1}$ on the complement of $Z$.
        We show in Appendix \ref{Appendix: coefficient function theorem} that $\gamma$ is continuous
        by showing that it is computable.
        This is essentially the continuous-time version of \cite[Corollary 4.4]{DPLSR20}.
        Then, since $Z$ is nowhere dense, it follows that $\gamma$ and $F_{m_1, 1,\dots, 1}$ coincide on 
        the complement of $Z'$.
        It follows that 
        \[
            F_{m_1,m_2,\dots,m_N} \left(\lambda_1,\dots,\lambda_N, u_0,\dots,u_{n - 1}\right)
            =
            F_{m_1,1,\dots,1}\left(\lambda_1, \lambda_2, \dots, \lambda_2, \dots, \lambda_N,\dots,\lambda_N, u_0,\dots, u_{n - 1}\right).
        \]
    \end{enumerate}
\end{proof}

Theorem \ref{Theorem: coefficient functions} has the following straightforward corollary, 
which is useful for constructing perturbations. It is proved in Appendix \ref{Appendix: Proposition relating coefficient functions}.

\begin{Proposition}\label{Proposition: coefficient functions for equal Cauchy problems}
    Let $m, n_1, n_2 \in \N$ with $m \leq n_2 \leq n_1$.
    Let $\lambda_1,\dots,\lambda_{n_1}, u_0,\dots, u_{n_1 - 1} \in \C$.
    Assume that the finite sequence $(u_k)_k$ satisfies the linear recurrence with 
    initial values $u_0,\dots,u_{n_2 - 1}$ and characteristic polynomial 
    $\prod_{j = 1}^{n_2} (z - \lambda_j)$.
    Then the number 
    \[
        G_{m, n_1}(\lambda_1, \lambda_{m + 1},\dots,\lambda_{n_1}, u_0,\dots,u_{n_1 - 1})
    \]
    is equal to 
    \[
        \left(
            \prod_{j = n_2 + 1}^{n_1} \left(\lambda_j - \lambda_1\right)
        \right)
        G_{m, n_2}(\lambda_1, \lambda_{m + 1},\dots,\lambda_{n_2}, u_0,\dots,u_{n_2 - 1}).
    \]
\end{Proposition}

%% file: Non-Equality.tex
\section{Proof of Theorem \ref{Theorem: equality checking}}

Our algorithm for Ultimate Inequality involves checking if the difference of the two given functions has a characteristic root $\lambda$ with non-zero imaginary part
such that $\lambda$ dominates all real characteristic roots and the coefficient of $e^{\lambda t}$ in \eqref{eq: exponential polynomial with indexes} is non-zero.
The complexity analysis of this step is closely related to a problem of independent interest: testing functions for equality.

The algorithm for testing two Cauchy problems $(c, u) \in C_n(\C)$ and $(d, v) \in C_m(\C)$ for equality is straightforward:
compute the first $n + m$ terms 
$u_0,\dots,u_{n + m - 1}$ and $v_0,\dots,v_{n + m - 1}$ 
of the linear recurrence sequences specified by $(c, u)$ and $(d, v)$
and halt if there exists 
$0 \leq j \leq n + m - 1$
such that $u_j \neq v_j$.
This algorithm is correct since $\sem{(c, u)} - \sem{(d, v)}$ satisfies a differential equation of order $n + m$ whose solutions for a given vector of initial values must be unique.
It is not obvious, however, that this algorithm runs in polynomial time.
In order to show this one needs to prove that if all of the differences 
$|u_j - v_j|$
are small, then there exist small perturbations of $(c, u)$ and $(d, v)$ 
such that the perturbed differences become equal.
Perhaps surprisingly, there does not appear to be an immediately obvious way of constructing such a perturbation.

We can ``extract'' the required construction from an elementary constructive proof of the fact that if the first $n + m$ terms of the sequences $(u_j)_j$ and $(v_j)_j$ agree, then $(u_j)_j$ satisfies the linear recurrence with characteristic polynomial $\chi_d$. 
The proof is based on the following simple lemmas: 
If a sequence satisfies a recurrence equation with characteristic polynomial $\chi$, then it satisfies the recurrence equation with characteristic polynomial $(z - \lambda) \chi$ for all $\lambda \in \C$.
If a sequence satisfies a recurrence equation with characteristic polynomials $(z - \alpha) \chi$ and $(z - \beta) \chi$, then it satisfies the recurrence equation with characteristic polynomial $\chi$.
If a sequence satisfies a recurrence equation with characteristic polynomial $\chi$ of degree $n$ and its first $n$ terms satisfy a recurrence equation with characteristic polynomial $\psi$ with $\psi | \chi$, then the whole sequence satisfies the recurrence equation with characteristic polynomial $\psi$. 
The proof can then be outlined as follows:
We will be able to assume without loss of generality that $n = m$.
Let $\lambda_1,\dots,\lambda_n$ denote the roots of $\chi_c$, let $\mu_1,\dots,\mu_n$ denote the roots of $\chi_d$.
Using the first lemma, it follows by induction that $(u_j)_j$ satisfies the equations 
$\left(\prod_{j = 1}^k (z - \mu_j)\right) \chi_c$
and that $(v_j)_j$ satisfies the equations
$\left(\prod_{j = 1}^k (z - \lambda_j)\right) \chi_d$ 
for $k = 0,\dots,n$.
Up to relabelling we can write 
$
\chi_d \chi_c = \left(\prod_{j = s + 1}^n (z - \lambda_j)\right) \left(\prod_{j = s + 1}^n(z - \mu_j)\right)\gcd(\chi_d, \chi_c)
$
for some $s \in \{0, \dots, n\}$.
We then show that the finite sequence $u_0,\dots,u_{2n - 1}$ satisfies the linear recurrence equation with characteristic polynomial 
$\gcd(\chi_c, \chi_d)$ by inductively employing the second lemma in a pattern that is illustrated below for the case $n - s = 2$:

\begin{center}
    \footnotesize
    \begin{tikzpicture}
        \node (21) at (0,1) {$(z - \lambda_1)(z - \lambda_2)(z - \mu_1)\gcd(\chi_d, \chi_c)$};
        \node (20) at (-2,0) {$\chi_c = (z - \lambda_1)(z - \lambda_2)\gcd(\chi_d, \chi_c)$};

        \node (12) at (6,1) {$(z - \lambda_1)(z - \mu_1)(z - \mu_2)\gcd(\chi_d, \chi_c)$};
        \node (02) at (8,0) {$\chi_d = (z - \mu_1)(z - \mu_2)\gcd(\chi_d, \chi_c)$};

        \node (11) at (3,0) {$(z - \lambda_1)(z - \mu_1)\gcd(\chi_d, \chi_c)$};

        \node (10) at (0.5,-1) {$(z - \lambda_1)\gcd(\chi_d, \chi_c)$};

        \node (01) at (5,-1) {$(z - \mu_1)\gcd(\chi_d, \chi_c)$};

        \node (00) at (3,-2) {$\gcd(\chi_d, \chi_c)$};

        \draw[->] (20) to node[near start, above left] {\footnotesize 1st Lemma} (21);
        \draw[->] (02) to node[near start, above right] {\footnotesize 1st Lemma} (12);
        \draw[->] (21) to node[near end, above right] {\hspace{0.3cm} \footnotesize 2nd Lemma} (11);
        \draw[->] (12) to (11);
        \draw[->] (20) to node[near end,above right] {\hspace{0.1cm} \footnotesize 2nd Lemma}(10);
        \draw[->] (11) to (10);
        \draw[->] (02) to (01);
        \draw[->] (11) to node[near end,above right] {\hspace{0.1cm} \footnotesize 2nd Lemma}(01);
        \draw[->] (10) to node[near end,above right] {\hspace{0.1cm} \footnotesize 2nd Lemma} (00);
        \draw[->] (01) to (00);
    \end{tikzpicture}
\end{center}
It then follows from the third lemma that $(u_j)_j$ satisfies the linear recurrence equation with characteristic polynomial $\gcd(\chi_c, \chi_d)$. The first lemma shows that it satisfies the recurrence equation with characteristic polynomial $\chi_d$.

The above proof may seem unnecessarily complex and overly detailed for the purpose of establishing the qualitative result, but it translates immediately to a proof of a quantitative version of the theorem. 
The lemmas involved in the proof can be easily strengthened to results that say that if the linear recurrence equations in the premises are satisfied ``approximately'', then the linear recurrence equations in the conclusion are satisfied ``approximately''.
By replacing the inductive arguments by recursive estimates we arrive at the next result.
Theorem \ref{Theorem: equality checking} follows immediately from this result.

\begin{Lemma}\label{Lemma: approximate equality perturbation}
    Let $(c, u, d, v) \in C_n(\C) \times C_m(\C)$.
    Let $(u_j)_j$ be the linear recurrence sequence with initial values $u$ and characteristic polynomial $\chi_c$.
    Let $(v_j)_j$ be the linear recurrence sequence with initial values $v$ and characteristic polynomial $\chi_d$.
    Assume that 
    $
        \left|
            u_j - v_j 
        \right|
        < \varepsilon < 1
    $
    for $j = \min\{n, m\}, \dots, n + m$.
    Then there exists a perturbation
        $(\widetilde{c}, \widetilde{u}, \widetilde{d}, \widetilde{v}) \in C_n(\C) \times C_m(\C)$
    of $(c, u, d, v)$ by at most  
    
    $
        O\left(
            \left(n + m + \norm{c}_{\infty} + \norm{d}_{\infty} + \norm{u}_\infty\right)^{O(\max\{n, m\}\log \max\{n, m\})}
        \right)
        \varepsilon^{\tfrac{1}{4\max\{n, m\} - 2}}
    $
    
    \noindent
    with respect to the spectral distance,
    such that 
    $\sem{\left(\widetilde{c}, \widetilde{u}\right)} = 
    \sem{\left(\widetilde{d}, \widetilde{v}\right)}$.
\end{Lemma}

Let us now prove Lemma \ref{Lemma: approximate equality perturbation}.
The next three lemmas are ``quantitative'' versions of the lemmas we used in the proof outline above.
The first two lemmas are proved by straightforward calculation.

\begin{Lemma}\label{Lemma: adding eigenvalue}
    Let 
    $u_0,\dots,u_{2n - 1} \in \C$ with 
    $\left|u_{k} + c_{n - 1} u_{k - 1} + \dots + c_0 u_{k - n}\right| < \varepsilon$
    for $k \geq n$.
    Let $\chi_d = (z - \alpha)\chi_c$.
    Then we have:
    \[
        \left|u_{k} + d_{n} u_{k - 1} + \dots + d_0 u_{k - n - 1}\right| < \left(1 + |\alpha|\right) \varepsilon.
    \]
    for $k \geq n + 1$.
\end{Lemma}

\begin{Lemma}\label{Lemma: eliminating non-shared eigenvalues}
    Let $\alpha, \beta, \lambda_1,\dots,\lambda_{n - 1}$ be complex numbers with $\alpha \neq \beta$.
    Consider the polynomials
    $
        \chi_c(z) = (z - \alpha) \prod_{j = 1}^{n - 1} (z - \lambda_j) = \sum_{k = 0}^{n - 1} c_k z^k
    $
    and 
    $
        \chi_d(z) = (z - \beta) \prod_{j = 1}^{n - 1} (z - \lambda_j) = \sum_{k = 0}^{n - 1} d_k z^k.
    $
    Let 
    $
        \chi_e(z) = \prod_{j = 1}^{n - 1} (z - \lambda_j) = \sum_{k = 0}^{n - 2} e_k z^k.
    $
    Let $u_0, \dots, u_n \in \R$ satisfy 
    $
        \left|u_n + \sum_{k = 0}^{n - 1} c_k u_k\right| < \varepsilon
    $
    and
    $
        \left|u_n + \sum_{k = 0}^{n - 1} d_k u_k\right| < \delta.
    $

    Then 
    $
        \left|
            u_{n - 1} + \sum_{k = 0}^{n - 1} e_k u_k
        \right|
        <
        \frac{\varepsilon + \delta}{|\alpha - \beta|}.
    $
\end{Lemma}

\begin{Lemma}\label{Lemma: perturbing equation to simpler equation}
    Let $(c, u) \in C_n$.
    Let $(e, w) \in C_m$ with $m \leq n$ 
    be such that there exists $e' \in \R^m$ 
    with $\chi_{e'} | \chi_c$ and:
    \begin{enumerate}
        \item $d_{\sigma}(e, e') < \varepsilon$.
        \item $|w_j - u_j| < \varepsilon$ for $j = 0, \dots, m - 1$.
        \item $|u_{k} + e'_{m - 1} u_{k - 1} + \dots + e'_{0} u_{k - m}| < \varepsilon$ for $k = m,\dots,n$.
    \end{enumerate}
    Then there exists 
    $(\widetilde{c}, \widetilde{u}) \in C_n$ with 
    \[ 
        d_{\sigma}\left((c, u), \left(\widetilde{c}, \widetilde{u}\right)\right) 
        =
        O\left(
            \left(n + \norm{c}_{\infty} + \norm{u}_\infty\right)^{O(n\log n)}
        \right)\varepsilon
    \]
    and 
    $\widetilde{u}_j = w_j$ for $j = 0,\dots m -1$
    such that 
    $\sem{(\widetilde{c}, \widetilde{u})} = \sem{(e,w)}$.
\end{Lemma}
\begin{proof}
    Write $\chi_c = \chi_{e'} \cdot \chi_{c'}$.
    Let 
    $\widetilde{c}$
    be defined by letting 
    $\chi_{\tilde{c}} = \chi_{e} \cdot \chi_{c'}$.

    Let $\widetilde{u}_j = w_j$ for $j = 0,\dots, m - 1$.
    For $j \geq m$, let $\widetilde{u}_j$ be defined by the recurrence
    \[\widetilde{u}_{j} = -e_{m - 1} \widetilde{u}_{j - 1} - \dots - e_{0} \widetilde{u}_{j - m}.\]
    
    By construction, the function $\sem{(\widetilde{c}, \widetilde{u})}$ defines the linear recurrence with characteristic polynomial $\chi_e$.
    Its first $m$ derivatives agree with those of $\sem{(e,w)}$.
    Hence we must have $\sem{(\widetilde{c}, \widetilde{u})} = \sem{(e,w)}$.

    By construction, the spectral distance between $\widetilde{c}$ and $c$ is smaller than $\varepsilon$
    and the distance between $\widetilde{u}_j$ and $u_j$ is smaller than $\varepsilon$ for $j = 0,\dots, m -1$.
    It remains to estimate the error 
    $|\widetilde{u}_j - u_j|$
    for $j = m, \dots, n$.
    Let $\varepsilon_j$ denote $|\widetilde{u}_j - u_j|$.
    We have $\varepsilon_j < \varepsilon$ for $j = 0, \dots, m -1$
    and 
    $\varepsilon_j \leq \varepsilon_{j + 1}$
    for $j = 0, \dots, m - 2$.

    We have:
    \begin{align*}
        |\widetilde{u}_j - u_j|
        &=
        \left|
        -e_{m - 1} \widetilde{u}_{j - 1} - \dots - e_{0} \widetilde{u}_{j - m}
        -
        u_j 
        \right|\\
        &\leq 
        \left|
            u_j + e'_{m - 1} u_{j - 1} + \dots + e'_{0} u_{j - m}
        \right|
        +
        \sum_{k = 1}^m
        \left|
            e'_{m - k} u_{j - k} - e_{m - k} \widetilde{u}_{j - k}
        \right| \\
        &\leq 
        \varepsilon
        +
        \sum_{k = 1}^m
        \left(
        \left|
            e'_{m - k}
        \right|
        \cdot
        \left|u_{j - k} - \widetilde{u}_{j - k}\right|
        +
        \left|
            \widetilde{u}_{j - k} 
        \right|
        \cdot 
        \left|
            e'_{m - k} - e_{m - k}
        \right|
        \right)\\
        &\leq 
        \varepsilon
        +
        \sum_{k = 1}^m
        \left(
        \norm{
            e' 
        }
        \cdot
        \left|u_{j - k} - \widetilde{u}_{j - k}\right|
        +
        \norm{
            \widetilde{u}
        }
        \cdot
        \norm{e' - e}
        \right)\\
        &\leq 
        \varepsilon
        +
        \norm{e'}
        m
        \varepsilon_{j - 1}
        +
        m
        \norm{
            \widetilde{u}
        }
        \cdot 
        \norm{e' - e}.
    \end{align*}

    It follows that 
    \[
        \varepsilon_{n} 
        \leq 
        \norm{e'}^{n - m + 1} m^{n - m + 1} \varepsilon
        + \left(m \norm{\widetilde{u}} \cdot \norm{e' - e} + \varepsilon\right) 
        \frac{1 - \norm{e'}^{n - m + 1}m^{n - m + 1}}{1 - \norm{e'}m}. 
    \]
    Estimating generously, we obtain 
    $\norm{e - e'} < n! (C + \varepsilon)^{n} \varepsilon$
    where $C$ is a bound on the roots of $\chi_c$.
    
    We obtain:
    \[
        \varepsilon_{n} 
        \leq 
        \left(\norm{e'}^{n - m + 1} m^{n - m + 1}
        + \left(m \norm{\widetilde{u}} \cdot n! (C + \varepsilon)^{n} + 1\right) 
        \frac{1 - \norm{e'}^{n - m + 1}m^{n - m + 1}}{1 - \norm{e'}m}
        \right)
        \varepsilon. 
    \]
    Standard estimates on $C$ and $\norm{e'}$ show that the bound is of the required shape.
\end{proof}

We are now in a position the prove Lemma \ref{Lemma: approximate equality perturbation}.
\begin{proof}[Proof of Lemma \ref{Lemma: approximate equality perturbation}]
First, assume that $n = m$.
Assume that $|u_k - v_k| \leq \varepsilon \leq 1$ for all $0 \leq k \leq n + m - 1$.
Then we have:
    \begin{align*}
        &\left|u_{k} + d_{n - 1} u_{k - 1} + \dots + d_{0} u_{k - n}\right|\\
        &\leq
            \left|u_{k} - v_k\right| 
            + \left|d_{n - 1}\right| \left|u_{k - 1} - v_{k - 1}\right| 
            + \dots 
            + \left|d_{0}\right| \left|u_{k - n} - v_{k - n}\right|\\
        &\leq 
        (n \norm{d}_{\infty} + 1) \varepsilon.
    \end{align*}
Write $\delta = \left(n \norm{d}_{\infty} + 1\right) \varepsilon$.
Write $\eta = \varepsilon^{\frac{1}{4n - 2}}$.
Write 
$\chi_c = \prod_{j = 1}^{n}(z - \lambda_j)$
and
$\chi_d = \prod_{j = 1}^{n}(z - \mu_j)$.

Consider the graph whose vertices are pairs of integers 
$(j, k)$ 
with 
$j \in \{1,2\}$ 
and 
$k \in \{1,\dots,n\}$,
and where there is an edge between 
$(1, k)$ 
and 
$(2, \ell)$
if and only if we have
$
    \left|\lambda_k - \mu_\ell\right| < \eta
$.
This graph is clearly bipartite.
Choose a maximal matching $\mathcal{M}$ in this graph.
Consider the set of vertices $\mathcal{V}$ which are not incident at an edge of $\mathcal{M}$.
Since the graph is bipartite, the set $\mathcal{V}$ has even cardinality,
all vertices in $\mathcal{V}$ are independent, 
and we have 
$\left|\Set{(j, k) \in \mathcal{V}}{j = 1}\right| 
= 
\left|\Set{(j, k) \in \mathcal{V}}{j = 2}\right|$.

Thus, up to relabelling the $\lambda_j$'s and $\mu_j$'s there exists an $s \in \{0,\dots,n\}$
such that 
$
    \left|\lambda_j - \mu_k\right| > \eta
$
for all $j, k \in \{1,\dots,s\}$
and 
$
    \left|\lambda_j - \mu_j\right| \leq \eta
$
for all $j > s$.

Consider the vector 
$e \in \R^{n - s}$ 
with 
$\chi_e(z) = \prod_{j = s + 1}^{n} (z - \lambda_j)$.
We will show that 
\[
    \left|
        u_k + e_{n - s - 1} u_{k - 1} + \dots + e_{0} u_{k - n + s}
    \right|
    \leq
    \binom{2 s - 1}{s - 1} \frac{\left(1 + \norm{c}_{\infty}\right)^{s - 1} \delta}{\eta^{2s - 1}}
\]
for all $n - s \leq k \leq 2n - 1$.

For $p, q \in \{0,\dots,s\}$, let $C(p, q)$ denote the coefficient vector of linear recurrence with characteristic polynomial
$
    \prod_{j = 1}^{p}(z - \lambda_j) \prod_{j = 1}^q (z - \mu_j) \prod_{j = s + 1}^{n} (z - \lambda_j)
$.
Let $\dim(p, q) = p + q + n - s$.
Write
\[
    E(p, q) = \max_{k \geq \dim(p,q)} \left|u_{k} + C(p, q)_{\dim(p,q) - 1} u_{k - 1} + \dots + + C(p, q)_{0} u_{k - \dim(p,q)} \right|.
\]
Then our claim can be rephrased as 
$E(0, 0) \leq \binom{2 s - 1}{s - 1} \frac{\left(1 + \norm{c}_{\infty}\right)^{s - 1} \delta}{\eta^{2s - 1}}$.

We will prove this by recursively estimating $E(p, q)$ in terms of $E(p + 1, q)$ and $E(p, q + 1)$.
The next two lemmas provide the base case and the recursive estimate:

\begin{Lemma}\label{Lemma: E(p, q) one argument equals s}
    We have 
    $E(s, q) = 0$
    and 
    $E(p, s) \leq (1 + \norm{c}_{\infty})^p\delta$.
\end{Lemma}
\begin{proof}
    We have $E(s, q) = 0$ since the sequence $(u_k)_k$ satisfies the linear recurrence associated with $\chi_c$, 
    which divides $\chi_{C(s, q)}$ for all $q \geq 0$.

    We have $E(p, 0) = \delta$ by definition of $\delta$.
    By Lemma \ref{Lemma: adding eigenvalue} we have $E(p, s + 1) \leq (1 + \norm{c}_{\infty}) E(p, s)$.
    The claim now follows by induction.
\end{proof}

\begin{Lemma}\label{Lemma: E(p, q) recursive equation}
    Let $p, q \in \{0,\dots,s - 1\}$.
    Then we have
    $E(p, q) \leq \frac{1}{\eta} \left(E(p + 1, q) + E(p, q + 1)\right)$.
\end{Lemma}
\begin{proof}
    This follows immediately from the definition of $E(p, q)$ and Lemma \ref{Lemma: eliminating non-shared eigenvalues}.
\end{proof}

Finally, we can recursively estimate $E(p, q)$ for all $p$ and $q$:

\begin{Lemma}\label{Lemma: estimating E(p, q)}
    We have the following estimate for all $q, p \in \{0,\dots,s\}$:
    \[ 
        E(p, q) \leq \binom{2s - 1 - (p + q)}{s - 1 - p} \frac{\left(1 + \norm{c}_\infty\right)^{s - 1} \delta}{\eta^{2s - 1 - (p + q)}}.
    \]
    In the above formula, we let $\binom{n}{k} = 0$ for all $n \geq 0$ and all $k < 0$.
\end{Lemma}
\begin{proof}
    We prove the claim by induction on 
    $k = 2s - (p + q)$
    when $k \geq 1$.
    
    Consider the case $k = 1$.
    Then $(p,q) \in \{(s, s - 1), (s - 1, s)\}$.
    We have 
    \[ 
        E(s, s - 1) = 0 = \binom{0}{-1} \frac{\left(1 + \norm{c}_\infty\right)^{s - 1}\delta}{\eta^{0}}
    \]
    and 
    \[ 
        E(s - 1, s) 
        \leq 
        (1 + \norm{c}_{\infty})^{s - 1}\delta 
        =
        \binom{0}{0} \frac{\left(1 + \norm{c}_\infty\right)^{s - 1}\delta}{\eta^{0}}.
    \]
    
    Assume that we have established the claim for $k < 2s$.
    Consider $p, q$ with 
    $k + 1 = 2s - (p + q)$,
    \textit{i.e.}
    $p + q = 2s - (k + 1)$.

    If $p = s$, then $q = s - (k + 1) \geq 0$, so that
    \[
        E(p, q) 
        = 0 
        =
        \binom{k}{-1} \frac{\left(1 + \norm{c}_\infty\right)^{s - 1}\delta}{\eta^{k}}.
    \]
    If $q = s$, then $p = s - (k + 1) \geq 0$
    \[
        E(p, q) 
        \leq 
        \left(1 + \norm{c}_\infty\right)^{s - 1} \delta 
        \leq 
        \binom{k}{k} \frac{\left(1 + \norm{c}_\infty\right)^{s - 1} \delta}{\eta^{k}}.
    \]
    Here, we have used that $\eta < 1$.

    Now, assume that $p, q \in \{0,\dots,s - 1\}$.
    Then by Lemma \ref{Lemma: E(p, q) recursive equation}, we have 
    \[
        E(p, q) \leq \frac{1}{\eta} \left(E(p + 1, q) + E(p, q + 1)\right).
    \]
    Employing the induction hypothesis we obtain:
    \begin{align*}
        &E(p, q)\\ 
        &\leq 
        \frac{1}{\eta} 
            \left(
                \binom{2s - 1 - (p + 1 + q)}{s - 1 - (p + 1)} \frac{\left(1 + \norm{c}_\infty\right)^{s - 1} \delta}{\eta^{2s - 1 - (p + 1 + q)}}
                + 
                \binom{2s - 1 - (p + q + 1)}{s - 1 - p} \frac{\left(1 + \norm{c}_\infty\right)^{s - 1} \delta}{\eta^{2s - 1 - (p + q + 1)}}
            \right)
        \\
        &= \binom{2s - 1 - (p + q)}{s - 1 - p} \frac{\left(1 + \norm{c}_\infty\right)^{s - 1} \delta}{\eta^{2s - 1 - (p + q)}}
        .
    \end{align*}
\end{proof}

Our claim now follows by setting $p = q = 0$ in Lemma \ref{Lemma: estimating E(p, q)}.
Since we have $\eta = \varepsilon^{\frac{1}{4n - 2}}$, we obtain 
\[
    E(0, 0) \leq 
    \binom{2 s - 1}{s - 1} \left(1 + \norm{c}_{\infty}\right)^{s - 1}\left(n \norm{d}_{\infty} + 1\right) \varepsilon^{\tfrac{1}{2}}.
\]
Now, let $\chi_e = \prod_{j = s + 1}^n (z - \lambda_j)$.
By Lemma \ref{Lemma: perturbing equation to simpler equation} there exists 
$\left(\widetilde{c}, \widetilde{u}\right) \in C_n(\C)$
with $\sem{(\widetilde{c},\widetilde{u})} = \sem{(e, u)}$ and 
\[ 
    d_{\sigma}
        \left(
            (c,u),
            \left(\widetilde{c}, \widetilde{u}\right)
        \right)
    \leq
    O\left(
        \left(n + \norm{c}_{\infty} + \norm{d}_{\infty} + \norm{u}_\infty\right)^{O(n\log n)}
    \right)
    \varepsilon^{\tfrac{1}{2}}.
\]
Applying Lemma \ref{Lemma: perturbing equation to simpler equation} to $(d, v)$ and $(e, u)$ we obtain a perturbation
$\left(\widetilde{d}, \widetilde{v}\right)$ of $(d, v)$ 
with 
$\sem{\left(\widetilde{d}, \widetilde{v}\right)} = \sem{(e, u)}$
and 
\[
    d_{\sigma}
        \left(
            (d, v),
            \left(\widetilde{d}, \widetilde{v}\right)
        \right)
    \leq
    O\left(
        \left(n + \norm{c}_{\infty} + \norm{d}_{\infty} + \norm{u}_\infty\right)^{O(n\log n)}
    \right)
    \varepsilon^{\tfrac{1}{4n - 2}}.
\]

If $n \neq m$, say $m < n$, choose $n - m$ roots of $\chi_c$ and add them as roots to $\chi_d$, carry out the above perturbation,
and remove the added roots from the characteristic polynomial of the perturbed instance (observing that the added roots themselves do not get perturbed).
This yields an estimate of the same shape.
\end{proof}

%% file: Inequality-Algorithm.tex
We can now give the full algorithm for inequality testing:

\begin{Algorithm}\label{Algorithm: Inequality}
    \hfill
\begin{itemize}
    \item 
    \textbf{Input.} A pair of C-finite functions, specified by Cauchy problems $(c, u, d, v) \in C_n(\R) \times C_m(\R)$.
    \item $\textbf{Behaviour.}$ 
    The algorithm may halt and return a truth value or run indefinitely. 
    If the algorithm halts, it returns $\text{``true''}$ if and only if 
    $\sem{(c,u)}(t) \geq \sem{(d,v)}(t)$ for all sufficiently large $t \geq 0$.
    \item $\textbf{Procedure.}$
\begin{enumerate}
    \item\label{Step: computation of roots}
    Compute $\lambda_1,\dots,\lambda_{n + m} \in \C$ 
    such that the list $\lambda_1,\dots,\lambda_n$ contains all roots of $\chi_c$,
    and the list $\lambda_{n + 1},\dots,\lambda_{n + m} \in \C$ contains all roots of $\chi_d$, 
    listed with multiplicity.

   \item\label{Step: bound} Compute an integer 
    $B > \log\left(\max\left\{\norm{u}_\infty, \norm{v}_\infty, |\lambda_1|,\dots,|\lambda_{n + m}|\right\} + 1\right)$. 
    
    \item For $N \in \N$:

    \begin{enumerate}[3.1.]
        \item 
    Let $M = \max\{\Omega\left(n, B, N + 1\right), \Omega\left(m, B, N + 1\right)\}$,
    where $\Omega$ is the polynomial from the second item of Theorem \ref{Theorem: coefficient functions}.

    \item Query the numbers $\lambda_j$ for approximations $\widetilde{\lambda}_j \in \Q[i]$ with 
        $|\Re(\widetilde{\lambda}_j) - \Re(\lambda_j)| < 2^{-M}$ and $|\Im(\widetilde{\lambda}_j) - \Im(\lambda_j)| < 2^{-M}$
    such that the polynomials 
        $\prod_{j = 1}^n \left(z - \widetilde{\lambda}_j\right)$ and $\prod_{j = n + 1}^{n + m} \left(z - \widetilde{\lambda}_j\right)$
    have real coefficients.

    \item Compute an $(n + m) \times (n + m)$-matrix encoding the relation
    $\preceq_M \subseteq \{1,\dots,n + m\}$ which is defined as follows:
            $j \preceq_M k$
        if and only if 
            $\Re(\widetilde{\lambda}_j) - 2^{-M} < \Re(\widetilde{\lambda}_k) + 2^{-M}$.

    \item Compute the sets 
        $\mathcal{M}_1 = \Set{k \in \{1,\dots,n\}}{j \preceq_M k \text{ for all }j \in \{1,\dots,n + m\}}$
        and

        $\mathcal{M}_2 = \Set{k \in \{n + 1,\dots,n + m\}}{j \preceq_M k \text{ for all }j \in \{1,\dots,n + m\}}$

    \item Initialise Kleeneans\footnote{\textit{i.e.}, variables that can assume three values: \texttt{true}, \texttt{false}, and \texttt{unknown}.}  
    $\texttt{c-positive?}$ and $\texttt{d-positive?}$ with value $\texttt{unknown}$.

    \item If $|\mathcal{M}_1| = 1$:
    
        Writing $\mathcal{M}_1 = \{j_1\}$ and $\{j_2,\dots,j_n\} = \{1,\dots,n\}\setminus\{j_1\}$,
        compute an approximation to 
        $(-1)^{n - 1} G_{1,n}(\lambda_{j_1},\dots,\lambda_{j_n}, u_0,\dots,u_{n - 1})$
        to error $2^{-N - 1}$.
        If the result is greater than or equal to $2^{-N}$,
            assign the value $\texttt{true}$ to $\texttt{c-positive?}$.
        If the result is less than or equal to $-2^{-N}$,
            assign the value $\texttt{false}$ to $\texttt{c-positive?}$.
        
    \item If $|\mathcal{M}_2| = 1$:
    
        Writing $\mathcal{M}_2 = \{k_1\}$ and $\{k_2,\dots,k_n\} = \{n + 1,\dots,n + m\}\setminus\{k_1\}$,
        compute an approximation to 
        $(-1)^{m - 1} G_{1,m}(\lambda_{k_1},\dots,\lambda_{k_n}, v_0, \dots, v_{m - 1})$
            to error $2^{-N - 1}$.
        If the result is greater than or equal to $2^{-N}$,
            assign the value $\texttt{true}$ to $\texttt{d-positive?}$.
        If the result is less than or equal to $-2^{-N}$,
            assign the value $\texttt{false}$ to $\texttt{d-positive?}$.
    
    \item\label{Step: halt if c has dominant real root} If $|\mathcal{M}_1| = 1$, $|\mathcal{M}_2| = 0$, and $\texttt{c-positive?} \neq \texttt{unknown}$:
        Halt and output the value of $\texttt{c-positive?}$.
    \item\label{Step: halt if d has dominant real root} If $|\mathcal{M}_1| = 0$ and $|\mathcal{M}_2| = 1$, and $\texttt{d-positive?} \neq \texttt{unknown}$:
        Halt and output the negated value of $\texttt{d-positive?}$.
    \item\label{Step: halt if c and d have dominant real roots} If $|\mathcal{M}_1| = 1$, $|\mathcal{M}_2| = 1$, and $\texttt{c-positive?} = \lnot \texttt{d-positive?}$:
        Halt and output the value of $\texttt{c-positive?}$.
    \item 
    Compute the set $\mathcal{R}$ of indexes $j \in \{1,\dots,n+m\}$ such that $\Im(\widetilde{\lambda}_j) < 2^{-M}$. 

    \item Compute the sets 
        $\mathcal{MR}_1 = \Set{k \in \mathcal{R} \cap \{1,\dots,n\}}{j \preceq_M k \text{ for all }j \in \mathcal{R}}$
    and 

        $\mathcal{MR}_2 = \Set{k \in \mathcal{R} \cap \{n + 1,\dots,n + m\}}{j \preceq_M k \text{ for all }j \in \mathcal{R}}$.

    \item\label{Step: compute C} Compute the set 
        $\mathcal{C} = \Set{j \in \{1,\dots,n+m\}}{j \not\preceq_M k \text{ for all }k \in \mathcal{R}}$.
    \item If $\mathcal{C}$ is non-empty:
    \begin{enumerate}[3.{14}.1]
        \item 
        Compute the coefficients $e \in \R^{n + m}$ of the polynomial
            $\chi_e = \prod_{j \in \{1,\dots,n + m\}}(z - \lambda_j)$.
        \item 
        Compute the first $n + m$ terms of the sequence $w_j = u_j - v_j$.
        \item 
        Compute the coefficients $e' \in \C^\ell$ of the polynomial 
            $\chi_{e'} = \prod_{j \in \{1,\dots,n + m\} \setminus \mathcal{C}}(z - \lambda_j)$.
        \item 
        Let $(w'_j)_j$ be the recurrence sequence with initial values 
            $w'_j = w_j$ for $j = 0, \dots, \ell - 1$
            and characteristic polynomial $\chi_{e'}$.
        \item 
        Compute rational approximations $\widetilde{\varepsilon}_j$ to 
            $\varepsilon_j = |w_j - w'_j|$ for $j = \ell,\dots,n + m - 1$ to error $2^{-M}$.
        \item\label{Step: halt if there is dominant complex root}
            If $\widetilde{\varepsilon}_j > 2^{-M}$ for some $j$, halt and output $\textbf{false}$.
    \end{enumerate}
    
    \item Initialise an empty list $L = \left\langle\right\rangle$. 
    \item Let $m_1 = |\mathcal{MR}_1|$. Let $m_2 = |\mathcal{MR}_2|$.
    \item If $m_1 > 0$:
    \begin{enumerate}[3.{17}.1]
        \item 
        Pick an arbitrary index $j_1 \in \mathcal{MR}_1$.
        \item 
        Let $\{j_{m_1 + 1},\dots,j_n\} = \{1,\dots,n\}\setminus\{j_1\}$
        \item 
        For $\ell = 1,\dots,m_1$,
            Compute a rational approximation $\alpha_\ell$ to 
            
            $
                    (-1)^{n - \ell} 
                \Re G_{\ell,n}
                \left(
                    \lambda_{j_1},
                    \underbrace{\lambda_{j_1},\dots,\lambda_{j_1}}_{m_1 - \ell  \text{ times}},
                    \lambda_{j_{m_1+1}},\dots,\lambda_{j_n},
                    u_0,\dots,u_{n - 1}
                \right)
            $

            to error $2^{-N - 1}$.
            Add $\alpha_\ell$ to $L$.
    \end{enumerate}
    \item If $m_2 > 0$:
    \begin{enumerate}[3.{18}.1]
        \item Pick an arbitrary index $k_1 \in \mathcal{MR}_2$.
        \item Let $\{k_{m_2 + 1},\dots,k_m\} = \{n + 1,\dots, n + m\}\setminus\{k_2\}$
        \item For $\ell = 1,\dots,m_2$,
            Compute a rational approximation $\alpha_\ell$ to 
                
            $(-1)^{m - \ell + 1} 
                \Re G_{\ell,m}
                \left(
                    \lambda_{k_1},
                    \underbrace{\lambda_{k_1},\dots,\lambda_{k_1}}_{m_2 - \ell  \text{ times}},
                    \lambda_{k_{m_2+1}},\dots,\lambda_{k_m},
                    v_0,\dots,v_{m - 1}
                \right)$
                
            to error $2^{-N - 1}$.
            Add $\alpha_\ell$ to $L$.
    \end{enumerate}
    \item\label{Step: halt if largest real roots have negative coefficient} 
        If all elements of $L$ are strictly smaller than $-2^{-N}$: 
        halt and output $\textbf{false}$.
    \end{enumerate}
    \end{enumerate}
\end{itemize}
\end{Algorithm}

\begin{Remark}
    \begin{enumerate}
        \item 
        Our formulation of 
        Algorithm \ref{Algorithm: Inequality} requires an explicit evaluation of the ``modulus of continuity'' $\Omega$ of 
        the functions $G_{\ell, k}$.
        This can be avoided by evaluating the functions $G_{\ell, k}$ on suitable interval approximations of the input 
        using interval arithmetic.
        The only advantage of the formulation chosen here is that the complexity analysis becomes more straightforward.
    \item The relation $\preceq_M$ is an overapproximation of the relation 
    $j \preceq k$ defined by $\Re(\lambda_j) \leq \Re(\mu_k)$.
    Consequently, the set $\mathcal{M}_1$ contains all indexes of all roots of $\chi_c$ which are dominant roots of 
    $\chi_c \cdot \chi_d$.
    Similarly, if $\chi_c$ has real roots, then the set $\mathcal{MR}_1$ contains all indexes $j \in \{1,\dots,n\}$ such that $\lambda_j$ is the largest real root of 
    $\chi_c \cdot \chi_d$.
    Analogous statements hold true for the sets $\mathcal{M}_2$ and $\mathcal{MR}_2$.
    Dually, all indexes in the set $\mathcal{C}$ correspond to roots of $\chi_c \cdot \chi_d$ with non-zero imaginary part whose real part is strictly greater than any real root of $\chi_c \cdot \chi_d$.
    \end{enumerate}
\end{Remark}

%% file: inequality-partial-correctness.tex
Write $f = \sem{(c, u)}$ and $g = \sem{(d, v)}$.
We claim that if Algorithm \ref{Algorithm: Inequality} halts, then it reports correctly whether $f(t) \geq g(t)$ for all large $t$.

Consider the case where the algorithm halts in \hyperref[Step: halt if c has dominant real root]{Step 3.8}. 
Then the polynomial $\chi_c$ has a simple real root $\rho$ with $\rho > \Re(\nu)$ for all roots $\nu$ of $\chi_d$ and of $\chi_c/(z - \rho)$.
Without loss of generality, we have $\rho = \lambda_1$.
By Theorem \ref{Theorem: coefficient functions} we have 
\[
    f(t) - g(t) = F_{1, n}(\lambda_1, \dots,\lambda_n, u_0,\dots,u_{n - 1}) e^{\rho t} + o(e^{\rho t}).
\]
By definition, the sign of 
$F_{1, n}(\lambda_1,\dots,\lambda_n, u_0,\dots,u_{n - 1})$ 
is equal to 
\[
    \prod_{j = 2}^n\operatorname{sgn}(\lambda_j - \lambda_1) \operatorname{sgn}\left(G_{1, n}(\lambda_1, \dots,\lambda_n, u_0,\dots,u_{n - 1})\right).
\]
Since $\lambda_1$ is greater than all real roots of $\chi_c$, and the complex roots come in conjugate pairs, we have 
$\prod_{j = 2}^n\operatorname{sgn}(\lambda_j - \lambda_1) = (-1)^{n - 1}$.
It follows that the sign of $f(t) - g(t)$ is eventually equal to the sign of 
$(-1)^{n - 1} G_{1, n}(\lambda_1, \dots,\lambda_n, u_0,\dots,u_{n - 1})$,
as long as the latter is non-zero.
The correctness of the algorithm now follows easily.
The cases where the algorithm halts 
in Steps
\hyperref[Step: halt if d has dominant real root]{3.9}
or 
\hyperref[Step: halt if d has dominant real root]{3.10}
are treated analogously.

Assume that the algorithm halts in 
\hyperref[Step: halt if there is dominant complex root]{Step 3.14.6}.
Then there exists $j \in \{\ell,\dots,n + m\}$ such that $w_j \neq w'_j$.
Hence, $\sem{(e, w)} \neq \sem{(e',w')}$.
Observe that if $\lambda_j$ is a real eigenvalue of $(e, w)$, then $j \in \mathcal{R}$.
Further observe that if $j \in \mathcal{C}$, then $\Re(\lambda_j) > \Re(\lambda_k)$ for all $k \in \mathcal{R}$,
so that $\lambda_j \notin \R$ and $\Re(\lambda_j) > \lambda_k$ for all real eigenvalues $\lambda_k$ of $(e, w)$.

Now, if the coefficient of $\lambda_j$ in the exponential polynomial representation of $f - g$ is zero for all $j \in \mathcal{C}$,
then $\sem{(e, w)} = f - g = \sem{(e', w')}$ by construction of $\sem{(e', w')}$.
But we have $\sem{(e', w')} \neq \sem{(e,w)}$, so that $f - g$ is non-zero and does not have a real dominant characteristic root.
It follows from Corollary \ref{Corollary: C-finite function without real characteristic roots} that $f$ cannot be eventually greater than $g$.

Finally, assume that the algorithm halts in 
\hyperref[Step: halt if largest real roots have negative coefficient]{Step 3.19}.
Observe that we must have $\sem{(c, u)} \neq \sem{(d, v)}$ in this case. If $\chi_d$ and $\chi_c$ do not have real roots then it follows from $\sem{(c, u)} \neq \sem{(d, v)}$ and Corollary \ref{Corollary: C-finite function without real characteristic roots} that $f$ is not eventually greater than or equal to $g$.

Thus, let us assume that $\chi_c$ or $\chi_d$ have real roots.
Let $\rho$ denote the largest real element of $\sigma_{(c,u)} \cup \sigma_{(d,v)}$.
Assume that $\rho$ is a root of both $\chi_c$ and $\chi_d$.
The cases where $\rho$ is a root of only one the polynomials are similar.
Observe that if $j \in \{1,\dots,n\}$ with $\lambda_j = \rho$, then $j \in \mathcal{MR}_1$.
Similarly, if $j \in \{n + 1,\dots, n + m\}$ with $\lambda_j = \rho$, then $j \in \mathcal{MR}_2$.
In particular, $m_1$ is an upper bound on the multiplicity $\mu_1$ of $\rho$ as a root of $\chi_c$
and $m_2$ is an upper bound to the multiplicity $\mu_2$ of $\rho$ as a root of $\chi_d$.

By assumption, there exist $j_1 \in \mathcal{MR}_1$ and $k_1 \in \mathcal{MR}_2$ such that
\[
    (-1)^{n - \mu_1} G_{\mu_1,n}(\lambda_{j_1},\underbrace{\lambda_{j_1},\dots,\lambda_{j_1}}_{\mu_1 - mu_1 \text{ times}},\lambda_{j_{m_1+1}},\dots,\lambda_{j_n})
    < -2^{-N-1}
\]
and 
\[
    (-1)^{m - \mu_2 + 1} G_{\mu_2,m}(\lambda_{k_1},\underbrace{\lambda_{k_1},\dots,\lambda_{k_1}}_{\mu_2 - m_2 \text{ times}},\lambda_{k_{m_2+1}},\dots,\lambda_{k_m})
    < -2^{-N-1}.
\]
By construction, we have $|\rho - \lambda_j| < 2^{-M}$ for all $j \in \mathcal{MR}_1 \cup \mathcal{MR}_2$.
It follows that the leading coefficients of the polynomials 
$P_{\rho}(c,u,t)$ and $P_{\rho}(d,v,t)$ in \eqref{eq: exponential polynomial with indexes} are strictly negative.
Together with Corollary \ref{Corollary: C-finite function without real characteristic roots} the claim follows.

%% file: inequality-termination.tex
Let us now show that Algorithm \ref{Algorithm: Inequality} halts in polynomial time on all robust instances.
The roots of $\chi_c$ and $\chi_d$ can be computed in polynomial time \cite[p. 117]{KoBook}.
It is relatively easy to see that each iteration of the For-loop takes polynomially many steps.
It therefore suffices to show that the number of iterations of the For-loop is bounded polynomially in the negative logarithm of the spectral distance of $\left(c, u, d, v\right)$ to the set of boundary instances of Ultimate Inequality.

Assume that the algorithm does not halt within the first $N$ iterations of the For-loop.
We will show that there exist polynomially controlled perturbations of $(c, u, d, v)$
such that one of the perturbed instances is a ``Yes''-instance and the other is a ``No''-instance.

First consider the case where we have $|\mathcal{M}_1| = 1$ and $|\mathcal{M}_2| = 0$.
In this case we have 
\[
    f(t) - g(t) = F_{1, n}(\lambda_{j_1},\dots,\lambda_{j_n}, u_0, \dots, u_{n - 1}) e^{\lambda_{j_1} t} + o(e^{\lambda_{j_1} t}).
\]
Since the algorithm does not halt in \hyperref[Step: halt if c has dominant real root]{Step 3.8},
$
    \left|
        G_{1, n}(\lambda_{j_1},\dots,\lambda_{j_n}, u_0, \dots, u_{n - 1})
    \right|
$
is less than $2^{-N}$.

It follows from the definition of $G_{1, n}$ (Theorem \ref{Theorem: coefficient functions}) that 
there exists a perturbation of $u_{n - 1}$ by at most $2^{-N}$ such that 
$(-1)^{n - 1} G_{1, n}(\lambda_{j_1},\dots,\lambda_{j_n}, u_0, \dots, u_{n - 1})$ is strictly positive, and 
another perturbation by at most $2^{-N}$ such that it is strictly negative.
It is easy to see that the sign of $F_{1, n}(\lambda_{j_1},\dots,\lambda_{j_n}, u_0, \dots, u_{n - 1})$ is equal to that of $(-1)^{n - 1} G_{1, n}(\lambda_{j_1},\dots,\lambda_{j_n}, u_0, \dots, u_{n - 1})$.
Hence, these perturbations yield a ``Yes''-instance and a ``No''-instance of Ultimate Inequality.
The cases where $|\mathcal{M}_1| = 0$ and $|\mathcal{M}_2| = 1$
or $|\mathcal{M}_1| = 1$ and $|\mathcal{M}_2| = 1$ are treated analogously.

It remains to consider the case where $|\mathcal{M}_1| \geq 2$ or $|\mathcal{M}_2| \geq 2$.
Let us assume without loss of generality that $|\mathcal{M}_1| \geq 2$.

Let us first construct the ``No''-instance.
Consider the set of all $\lambda_j$ with $j \in \mathcal{M}_1$. 
If this set contains only real numbers, then the ball of radius $2^{-M}$ about the largest real root of $\chi_c \cdot \chi_d$ contains at least two real roots of $\chi_c$, counted with multiplicity.
It follows that there exist perturbations $\widetilde{c}$ and $\widetilde{d}$ of $c$ and $d$ by at most $2^{-M}$ such that all roots of 
$\chi_{\widetilde{c}} \cdot \chi_{\widetilde{d}}$ have non-zero imaginary part.
We can further ensure that all roots of $\chi_{\widetilde{c}} \cdot \chi_{\widetilde{d}}$  with maximal real part are roots of  $\chi_{\widetilde{c}}$.
Let $\lambda$ be a root of $\chi_{\widetilde{c}}$ with maximal real part.
Let $P_{\lambda} \in \C[t]$ denote its coefficient in the exponential polynomial solution of $(\widetilde{c}, u)$.
Up to an arbitrarily small perturbation, $P_{\lambda}$ is non-zero.
This induces an arbitrarily small perturbation of $u$ by \eqref{eq: Vandermonde equation}.
It then follows from Corollary \ref{Corollary: C-finite function without real characteristic roots} that $\sem{(\widetilde{c}, \widetilde{u})} - \sem{(\widetilde{d}, v)}$ assumes negative values at arbitrarily large times, so that $(\widetilde{c}, \widetilde{u}, \widetilde{d}, v)$ is a ``No''-instance of Ultimate Inequality.

Let us now construct the ``Yes''-instance.
Assume that the set $\mathcal{C}$ computed in \hyperref[Step: compute C]{Step 3.14} is non-empty.
Since the algorithm does not halt in \hyperref[Step: halt if there is dominant complex root]{Step 3.14.6}, 
it follows exactly like in the proof of Lemma \ref{Lemma: approximate equality perturbation} that there exists a perturbation
$(\widetilde{c}, \widetilde{u}, \widetilde{d}, \widetilde{v})$ of $(c, u, d, v)$ by at most 
\[
    \delta 
    :=
    O\left(
        \left(n + m + \norm{c}_{\infty} + \norm{d}_{\infty} + \norm{u}_\infty\right)^{O(\max\{n, m\}\log \max\{n, m\})}
    \right)
    2^{\tfrac{-M}{4\max\{n, m\} - 2}}
\]
such that we have 
$
    \sem{(\widetilde{c},\widetilde{u})} = 
        \sum_{j \in \{1,\dots,n\}}
        P_j(t) e^{\widetilde{\lambda}_j t}
$
and 
$
    \sem{(\widetilde{d},\widetilde{v})} = 
    \sum_{j \in \{n + 1,\dots,n + m\}}
    P_j(t) e^{\widetilde{\lambda}_j t}
$
with $|\Re(\lambda_j) - \Re(\widetilde{\lambda}_j)| < \delta$,
$|\Im(\lambda_j) - \Im(\widetilde{\lambda}_j)| < \delta$,
and $P_j = 0$ for all $j \in \mathcal{C}$.
In particular, the real part of the dominant characteristic roots of $\sem{(\widetilde{c},\widetilde{u})} - \sem{(\widetilde{d},\widetilde{v})}$ is $\delta$-close to the real part of a characteristic root whose imaginary part is at most $2\delta$.
Further, we can ensure that 
$\gcd(\chi_c, \chi_{\widetilde{c}}) = \left(\prod_{j \in \{1,\dots,n\}\setminus\mathcal{C}} \left(z - \lambda_j\right)\right)$
and 
$\gcd(\chi_d, \chi_{\widetilde{d}}) = \left(\prod_{j \in \{n + 1,\dots,n + m\}\setminus\mathcal{C}} \left(z - \lambda_j\right)\right)$.
If $\mathcal{C}$ is empty, then these properties already hold true for $(\widetilde{c}, \widetilde{u}, \widetilde{d}, \widetilde{v}) = (c, u, d, v)$.

Now, let $c' \in \R^{n'}$ and $d' \in \R^{m'}$ be defined by letting their characteristic polynomials be
$\chi_{c'} = \prod_{j \in \{1,\dots,n\} \setminus \mathcal{C}} \left(z - \widetilde{\lambda}_j\right)$
and
$\chi_{d'} = \prod_{j \in \{n + 1,\dots,n + m\} \setminus \mathcal{C}} \left(z - \widetilde{\lambda}_j\right)$.
Define initial values
$u' = \left(\widetilde{u}_0, \dots, \widetilde{u}_{n' - 1}\right)$
and 
$v' = \left(\widetilde{v}_0, \dots, \widetilde{v}_{m' - 1}\right)$.
Then we have 
$\sem{\left(c',u'\right)} = \sem{\left(\widetilde{c}, \widetilde{u}\right)}$
and 
$\sem{\left(d', v'\right)} = \sem{\left(\widetilde{d}, \widetilde{v}\right)}$.

Let $N'$ be the largest integer with 
$
    2^{-\max\{\Omega(n, B, N' + 1), \Omega(m, B, N' + 1)\}} > \delta. 
$
Then, since $\Omega$ is a polynomial and $-\log \delta$ depends polynomially on $N$, $n$, and $B$ 
we have $N' \geq \alpha(n + m + B) N^{1/\beta(n + m + B)} - \gamma(n + m + B)$ for polynomials $\alpha, \beta, \gamma \in \N[x]$.

Observe that $\mathcal{MR}_1 \cup \mathcal{MR}_2 \neq \emptyset$.
Since the algorithm does not halt, there exists -- without loss of generality -- 
an $\ell \in \{1,\dots,m_1\}$ such that 

$
    (-1)^{n - \ell} \Re G_{\ell,n}
        \left(
            \lambda_{j_1},
            \underbrace{\lambda_{j_1},\dots,\lambda_{j_1}}_{m_1 - \ell  \text{ times}},
            \lambda_{j_{m_1+1}},\dots,\lambda_{j_n},
            u_0,\dots,u_{n - 1}
        \right)
    \geq
    -2^{-N'}.
$

By construction, there exist 
$\widetilde{\lambda}_{j_1}, 
\widetilde{\lambda}_{j_{\ell + 1}},
\dots,
\widetilde{\lambda}_{j_{m_1}}$,
whose real and imaginary parts are $2^{-M}$-close to those of 
$\lambda_{j_k}$,
such that 
$\widetilde{\lambda}_{j_1}$ 
is real and 
$\Re(\widetilde{\lambda}_{j_1}) > \Re(\nu)$
for all 

$\nu \in \{\widetilde{\lambda}_{j_{\ell + 1}},\dots,\widetilde{\lambda}_{j_{m_1}},\lambda_{j_{m_1 + 1}},\dots,\lambda_{j_n}\}$.

We then have:
$
    G_{\ell, n}
    \left(
        \widetilde{\lambda}_{j_1},
        \widetilde{\lambda}_{j_{\ell + 1}},
        \dots,
        \widetilde{\lambda}_{j_{m_1}},
        \lambda_{j_{m_1+1}},
        \dots,
        \lambda_{j_{n'}},
        \lambda_{j_{n' + 1}}
        \dots,
        \lambda_{j_n},
        \widetilde{u}_0, \dots, \widetilde{u}_{n - 1}
    \right)
    \geq 
    -2^{-N' + 1}.
$

Further, up to perturbing the characteristic roots of $\chi_{d'}$ by at most $2^{-M}$, we can ensure that 
$\Re(\widetilde{\lambda}_{j_1}) > \Re(\nu)$
for all characteristic roots $\nu$ of $\chi_{d'}$.

Up to relabelling, we may assume that 
$\mathcal{C} \cap \{1, \dots, n\} = \{j_{n' + 1},\dots,j_n\}$.
By Proposition \ref{Proposition: coefficient functions for equal Cauchy problems} we have that 
\[
    G_{\ell, n}
    \left(
        \widetilde{\lambda}_{j_1},
        \widetilde{\lambda}_{j_{\ell + 1}},
        \dots,
        \widetilde{\lambda}_{j_{m_1}},
        \lambda_{j_{m_1+1}},
        \dots,
        \lambda_{j_{n'}},
        \lambda_{j_{n' + 1}}
        \dots,
        \lambda_{j_n},
        \widetilde{u}_0, \dots, \widetilde{u}_{n - 1}
    \right)
\]
is equal to 
\[
    \left(
            \prod_{k = n' + 1}^{n} \left(\lambda_{j_k} - \widetilde{\lambda}_{j_1}\right)
    \right)
    G_{\ell, n'}
    \left(
        \widetilde{\lambda}_{j_1},
        \widetilde{\lambda}_{j_{\ell + 1}},
        \dots,
        \widetilde{\lambda}_{j_{m_1}},
        \lambda_{j_{m_1+1}},
        \dots,
        \lambda_{j_{n'}}
        \widetilde{u}_0, \dots, \widetilde{u}_{n' - 1}
    \right)
\]
so that the number
\[
    (-1)^{n - \ell}
    G_{\ell, n'}
    \left(
        \widetilde{\lambda}_{j_1},
        \widetilde{\lambda}_{j_{\ell + 1}},
        \dots,
        \widetilde{\lambda}_{j_{m_1}},
        \lambda_{j_{m_1+1}},
        \dots,
        \lambda_{j_{n'}}
        \widetilde{u}_0, \dots, \widetilde{u}_{n' - 1}
    \right)
\]
is greater than or equal to
$
    \left(
            \prod_{k = n' + 1}^{n} \left(\lambda_{j_k} - \widetilde{\lambda}_{j_1}\right)
    \right)^{-1}
    -2^{-N' + 1}.
$
Now, 
$\left(
    \prod_{k = n' + 1}^{n} \left(\lambda_{j_k} - \widetilde{\lambda}_{j_1}\right)
\right)$
is a positive real number, since the numbers $\lambda_{j_k}$ with $k \geq n' + 1$ come in complex conjugate pairs.
Further, the difference $n - n'$ is even, so that $(-1)^{n - \ell} = (-1)^{n' - \ell}$.
Let $C$ be an upper bound on $\left|\lambda_{j_k} - \widetilde{\lambda}_{j_1}\right|$.
We obtain that
$
    (-1)^{n' - \ell}
    G_{\ell, n'}
    \left(
        \widetilde{\lambda}_{j_1},
        \widetilde{\lambda}_{j_{\ell + 1}},
        \dots,
        \widetilde{\lambda}_{j_{m_1}},
        \lambda_{j_{m_1+1}},
        \dots,
        \lambda_{j_{n'}}
        \widetilde{u}_0, \dots, \widetilde{u}_{n' - 1}
    \right)
$
is greater than or equal to 
$
-C^n2^{-N' + 1}.
$

Now,
by Proposition \ref{Proposition: coefficient functions for equal Cauchy problems},
the left-hand side of this inequality can be written as
\[
        \widetilde{u}_{n' - 1} 
        +
        \sum_{j = 1}^{n' - 1} (-1)^{n' + j} A_{n',\ell,j}
            \left(
                \widetilde{\lambda}_{j_1},
                \widetilde{\lambda}_{j_{\ell + 1}},
                \dots,
                \widetilde{\lambda}_{j_{m_1}},
                \lambda_{j_{m_1+1}},
                \dots,
                \lambda_{j_{n'}}
            \right) \widetilde{u}_{j - 1}
\]
so that there exists a perturbation 
$\widetilde{\widetilde{u}}_{n' - 1}$
of $\widetilde{u}_{n' - 1}$ by at most $C^n2^{-N' + 1}$ such that the 
number 

\[(-1)^{n' - \ell} G_{\ell, n'}
\left(
    \widetilde{\lambda}_{j_1},
    \widetilde{\lambda}_{j_{\ell + 1}},
    \dots,
    \widetilde{\lambda}_{j_{m_1}},
    \lambda_{j_{m_1+1}},
    \dots,
    \lambda_{j_{n'}}
    \widetilde{u}_0, \dots, \widetilde{u}_{n' - 2}, \widetilde{\widetilde{u}}_{n' - 1}
\right)
\]
is strictly positive.
Observing that $N'$ is controlled polynomially in the input data and that the logarithm of $C$ is bounded polynomially in the input data, we obtain that the perturbation is controlled polynomially.
We have thus constructed polynomially controlled perturbations 
$(\widetilde{c'}, \widetilde{u'}) \in C_{n'}(\R)$ and $(\widetilde{d'}, v') \in C_{m'}(\R)$
of 
$\left(c', u'\right)$ and $\left(d', v'\right)$
with 
$\sem{\left(\widetilde{c'}, \widetilde{u'}\right)}(t) \geq \sem{\left(\widetilde{d'}, v'\right)}(t)$ 
for all large $t$.
We can construct polynomially controlled perturbations of $(c, u)$ and $(d, v)$ by adding back the characteristic roots we have removed by 
passing from $c$ to $c'$ and from $d$ to $d'$ and by extending $\widetilde{u'}$ and $v'$ using the linear recurrence equations with characteristic polynomial $\chi_{\widetilde{c'}}$ and $\chi_{\widetilde{d'}}$ respectively.
By a calculation similar to Lemma \ref{Lemma: perturbing equation to simpler equation}, this yields polynomially controlled perturbations of $(c, u)$ and $(d, v)$.

%% file: Spectral-Distance.tex
\section{Proof of Proposition \ref{Proposition: spectral distance}}\label{Appendix: Spectral Distance}

We recall Bernstein's inequality:

\begin{Theorem}[Bernstein]\label{Theorem: Bernstein}
    Let $P$ be a complex polynomial of degree $n$.
    Then 
    \[
        \max_{|z| \leq 1}\left|P^{(k)}(z)\right|
        \leq 
        \frac{n!}{(n - k)!} \max_{|z| \leq 1} \left|P(z)\right|.
    \]
\end{Theorem}

\begin{Proposition}
    Let $c, c' \in \mathbb{K}^{n}$.
    Assume that $c \neq c'$.
    Then we have:
    \[
        - \log d_{\sigma}(c,c')
        \leq 
        - \log d(c,c')  + 2n \log (n) +  2n \log \left(\max\{\norm{c}_\infty, \norm{c'}_\infty\}\right).
    \]
\end{Proposition}
\begin{proof}
    Let $(\lambda_1,\dots,\lambda_n) \in \C^n$ be a vector containing all roots of the polynomial 
    $\chi_c$, where each root occurs according to its multiplicity.
    Let $(\mu_1,\dots,\mu_n) \in \C^n$ be a vector containing all roots of the polynomial $\chi_{c'}$, where each root occurs according to its multiplicity.
    Let $S_n$ denote the group of all bijections $\pi \colon \{1,\dots,n\} \to \{1,\dots,n\}$. 
    Let $\pi \colon \{1,\dots,n\} \to \{1,\dots,n\}$ be a bijection such that the minimum
    $\norm{(\lambda_1,\dots,\lambda_n) - (\mu_{\tau(1)},\dots,\mu_{\tau(n)})}_{\infty}$
    where $\tau$ ranges over all bijections of type $\{1,\dots,n\} \to \{1,\dots,n\}$
    is attained at $\pi$.
    Consider the sequence of polynomials
    \[
        P_j(z) = (z - \lambda_1)\cdot \dots \cdot (z - \lambda_j) \cdot (z - \mu_{\pi(j + 1)}) \cdot \dots \cdot (z - \mu_{\pi(n)}),
    \]
    for $j \in \{0,\dots,n\}$.
    Observe that 
    $P_0(z) = \chi_c(z)$
    and 
    $P_n(z) = \chi_{c'}(z)$.

    We have:
    \begin{align*}
        &\left|P_{j + 1}(z) - P_j(z)\right| \\
        &=
        \big|
            (z - \lambda_1)\cdot \dots \cdot (z - \lambda_j) \cdot (z - \lambda_{j + 1}) \cdot (z - \mu_{\pi(j + 2)}) \cdot \dots \cdot (z - \mu_{\pi(n)})
            \\
            &\hspace{1cm} - (z - \lambda_1)\cdot \dots \cdot (z - \lambda_j) \cdot (z - \mu_{\pi(j + 1)}) \cdot (z - \mu_{\pi(j + 2)}) \cdot \dots \cdot (z - \mu_{\pi(n)})
        \big|\\
        &=
        \left|
            (z - \lambda_1)\cdot \dots \cdot (z - \lambda_j) \cdot (\mu_{\pi(j + 1)} - \lambda_{j + 1}) \cdot (z - \mu_{\pi(j + 2)}) \cdot \dots \cdot (z - \mu_{\pi(n)})
        \right|\\
        &\leq d_{\sigma}(c,c') \cdot \left|(z - \lambda_1)\cdot \dots \cdot (z - \lambda_j)\cdot (z - \mu_{\pi(j + 2)}) \cdot \dots \cdot (z - \mu_{\pi(n)})\right|
    \end{align*}
    Write $C = \max\{\norm{c}_{\infty}, \norm{c'}_{\infty}\}$.
    By \cite[Proposition 5.9]{BasuPollackRoy} we have $|\lambda_j| \leq n C$ and $|\mu_j| \leq n C$ for all $j$. 
    We hence obtain:
    \[
        \left|P_{j + 1}(z) - P_j(z)\right| 
        \leq 
        d_{\sigma}(c,c') (n C)^{n - 1}
    \]
    for all $z$ in the complex unit disk.
    This implies 
    \begin{align*}
        \left| 
            \chi_c(z) - \chi_{c'}(z)
        \right|
        &\leq 
        \left|P_{n}(z) - P_{n - 1}(z)\right| 
        +
        \left|P_{n - 1}(z) - P_{n - 2}(z)\right| 
        +
        \dots
        +
        \left|P_{1}(z) - P_{0}(z)\right|\\
        &\leq 
        d_{\sigma}(c,c') n (n C)^{n - 1}
    \end{align*}
    for all $z$ in the complex unit disk.

    The coefficient of $z^k$ in the polynomial $\chi_c(z) - \chi_{c'}(z)$ is given by
    \[
        \frac{1}{k!}\left(\tfrac{d^k}{dz^k} \left(\chi_c(z) - \chi_{c'}(z)\right)\big|_{z = 0}\right).
    \]
    By Bernstein's inequality (Theorem \ref{Theorem: Bernstein}) we have 
    \[
        \left|\tfrac{d^k}{dz^k} \left(\chi_c(z) - \chi_{c'}(z)\right)\big|_{z = 0}\right| 
        \leq 
        \frac{n!}{(n - k)!} \max_{|z| \leq 1} \left| \chi_c(z) - \chi_d(z) \right|
        \leq 
        d_{\sigma}(c,c') (n C)^{2n}.
    \]
    It follows that 
    \[
        - \log d(c,c') \geq - \log d_{\sigma}(c,c') - 2 n \log (n C).
    \]
    This yields the result.
\end{proof}

%% file: Vandermonde-Lemmas.tex
\section{Proof of Lemma \ref{Lemma: recursive equations for Q}}
\label{Appendix: Vandermonde lemmas}

\begin{proof}[Proof of Lemma \ref{Lemma: recursive equations for Q}]
    \hfill
    \begin{enumerate}
        \item Observe that we have 
        $
        Y^k - Z^k =
            (Y - Z)\sum_{\ell_1 + \ell_2 = k - 1} Y^{\ell_1}Z^{\ell_2}
        $.
        We calculate:
        \begin{align*}
            &Q_{n, m, 0}(X_1,\dots,X_{n - 1}, Y) - Q_{n, m, 0}(X_1,\dots,X_{n - 1}, Z)\\
            &=  \sum_{k_1 + \dots + k_n = m}
                    \left(
                        Y^{k_n}X_{n - 1}^{k_{n - 1}} \cdot \dots \cdot X_1^{k_1}
                        -
                        Z^{k_n}X_{n - 1}^{k_{n - 1}} \cdot \dots \cdot X_1^{k_1}
                    \right)\\
            &=  \sum_{k_1 + \dots + k_n = m}
                    \left(
                        (Y^{k_n} - Z^{k_n})X_{n - 1}^{k_{n - 1}} \cdot \dots \cdot X_1^{k_1}
                    \right)\\
            &=  (Y - Z)
                \sum_{k_1 + \dots + k_n = m}
                    \left(
                        \left(\sum_{\ell_1 + \ell_2 = k_n - 1} Y^{\ell_1}Z^{\ell_2}\right)
                        X_{n - 1}^{k_{n - 1}} \cdot \dots \cdot X_1^{k_1}
                    \right)\\
            &=  (Y - Z)
            \sum_{k_1 + \dots + k_{n - 1} + \ell_1 + \ell_2 = m - 1}
                \left(
                    Y^{\ell_1}Z^{\ell_2}X_{n - 1}^{k_{n - 1}} \cdot \dots \cdot X_1^{k_1}
                \right)\\
            &= (Y - Z) Q_{n + 1, m - 1, 0}(X_1,\dots,X_{n - 1}, Y, Z).
        \end{align*}
        \item Let 
        \begin{align*}
            \Delta&(X_1,\dots,X_{n + 1}) = \\
            &(X_{n + 1} - X_1) Q_{n + 1, m - 1, i + 1}(X_1,\dots,X_{n + 1})\\
            &- Q_{n + 1, m , i}(X_1,\dots,X_{n + 1})\\
            &+ Q_{n, m , i + 1}(X_1,\dots,X_n).
        \end{align*}
        Our goal is to show that $\Delta$ is zero.
        Write out $\Delta$ according to the definition:
        \begin{align*}
            \Delta&(X_1,\dots,X_{n + 1}) =\\
            &\sum_{k_1 + \dots + k_{n + 1} = m - 1}\binom{k_1 + i + 1}{i + 1} X_{n+1}^{k_{n + 1} + 1} \cdot \dots \cdot X_1^{k_1} \\
            &- \sum_{k_1 + \dots + k_{n + 1} = m - 1}\binom{k_1 + i + 1}{i + 1} X_{n+1}^{k_{n + 1}} \cdot \dots \cdot X_1^{k_1 + 1} \\
            &- \sum_{k_1 + \dots + k_{n + 1} = m }\binom{k_1 + i}{i} X_{n+1}^{k_{n + 1}} \cdot \dots \cdot X_1^{k_1}\\
            &+ \sum_{k_1 + \dots + k_{n} = m }\binom{k_1 + i + 1}{i + 1} X_{n}^{k_{n}} \cdot \dots \cdot X_1^{k_1}.
        \end{align*}
        Re-write the first two sums in the above equation as follows:
        let $k_{n + 1}$ in the first sum range from $1$ to $m $ rather than from $0$ to $m - 1$.
        Let $k_1$ in the second sum range from $1$ to $m $ rather than from $0$ to $m - 1$.
        We obtain:
        \begin{align*}
            \Delta&(X_1,\dots,X_{n + 1}) =\\
            &\sum_{\substack{k_1 + \dots + k_{n + 1} = m  \\ k_{n + 1} \geq 1}}\binom{k_1 + i + 1}{i + 1} X_{n+1}^{k_{n + 1}} \cdot \dots \cdot X_1^{k_1} \\
            &- \sum_{\substack{k_1 + \dots + k_{n + 1} = m  \\ k_1 \geq 1}}\binom{k_1 + i + 1}{i + 1} X_{n+1}^{k_{n + 1}} \cdot \dots \cdot X_1^{k_1 + 1} \\
            &- \sum_{k_1 + \dots + k_{n + 1} = m }\binom{k_1 + i}{i} X_{n+1}^{k_{n + 1}} \cdot \dots \cdot X_1^{k_1}\\
            &+ \sum_{k_1 + \dots + k_{n} = m }\binom{k_1 + i + 1}{i + 1} X_{n}^{k_{n}} \cdot \dots \cdot X_1^{k_1}.
        \end{align*}
        To show that $\Delta(X_1,\dots,X_{n + 1})$ is zero, consider a monomial 
        $X_1^{k_1} \cdot \dots \cdot X_{n + 1}^{k_{n + 1}}$
        with $k_1 + \dots + k_{n + 1} = m $.
        We show that the coefficient of this monomial in $\Delta$ is zero.
        We distinguish four cases:
        \begin{enumerate}
            \item $k_1 \geq 1$, $k_{n + 1} \geq 1$. 
                In this case we get a contribution from the first three sums.
                The coefficient of $X_1^{k_1} \cdot \dots \cdot X_{n + 1}^{k_{n + 1}}$ is hence equal to 
                \[
                    \binom{k_1 + i + 1}{i + 1} - \binom{k_1 + i}{i + 1} - \binom{k_1 + i}{i} = 0.
                \]
            \item $k_1 \geq 1$, $k_{n + 1} = 0$. 
                In this case we get a contribution from the last three sums.
                The coefficient of $X_1^{k_1} \cdot \dots \cdot X_{n + 1}^{k_{n + 1}}$ is hence equal to 
                \[
                    -\binom{k_1 + i}{i + 1} - \binom{k_1 + i}{i} + \binom{k_1 + i + 1}{i + 1} = 0.
                \]
            \item $k_1 = 0$, $k_{n + 1} \geq 1$.
                In this case we get a contribution from the first and third sum. 
                The coefficient of $X_1^{k_1} \cdot \dots \cdot X_{n + 1}^{k_{n + 1}}$ is hence equal to 
                \[
                    \binom{k_1 + i + 1}{i + 1} - \binom{k_1 + i}{i}
                    =
                    \binom{i + 1}{i + 1} - \binom{i}{i}
                    =
                    0.
                \]
            \item $k_1 = 0$, $k_{n + 1} = 0$.
                In this case we get a contribution from the last two sums.
                The coefficient of $X_1^{k_1} \cdot \dots \cdot X_{n + 1}^{k_{n + 1}}$ is hence equal to 
                \[
                    -\binom{k_1 + i}{i} + \binom{k_1 + i + 1}{i + 1}
                    =
                    -\binom{i}{i} + \binom{i + 1}{i + 1}
                    = 0.
                \]
        \end{enumerate}
        Thus, $\Delta$ is indeed equal to zero, as claimed. This finishes the proof.
    \end{enumerate}
\end{proof}

Let us write
\[
    C^{n}_{k, m, i}(X_1,\dots,X_k) = 
        \left[
            Q_{k, m + p - n, i}(X_1,\dots,X_k)
        \right]_{p = 1,\dots,n}.
\]
In the above notation, $k$ denotes the number of variables,
$m$ denotes the degree of the last entry of the column,
and $i$ determines the binomial coefficient
$\binom{k_1 + i}{i}$.

Then the recursive equations from Lemma \ref{Lemma: recursive equations for Q} become:
\[
    C^n_{k, m, 0}(X_1,\dots,X_{k - 1}, Y) - C^n_{k, m, 0}(X_1,\dots,X_{k - 1}, Z)
    =
    (Y - Z) C^n_{k + 1, m - 1, 0}(X_1,\dots, X_{k - 1}, Y, Z)
\]
and 
\[
    C^n_{k + 1, m, i}(X_1,\dots,X_{k + 1}) - C^n_{k, m, i + 1}(X_1,\dots,X_{k})
    =
    (X_{k + 1} - X_1) C^n_{k + 1, m - 1, i + 1}(X_1,\dots,X_{k + 1}).
\]
These recursive equations allow us to decrease the degree of a column, which allows us 
to bring the matrix $\widetilde{V}_{m_1,1,\dots,1}(\Lambda_1,\dots,\Lambda_n)$ into lower echelon form.

\section{Minors of Generalised Vandermonde Matrices}
\label{Appendix: Vandermonde lemmas 2}

Based on Lemma \ref{Lemma: recursive equations for Q}
we can compute the minors of a generalised Vandermonde matrix:

\begin{Lemma}\label{Lemma: minors of generalised Vandermonde matrix}
    Let $m_1 \leq n$.
    Let $K = (\Q[i])(\Lambda_1,\Lambda_{m_1 + 1},\dots,\Lambda_n)$.

    \begin{enumerate}
        \item Assume that $m_1 = 1$.
        Then we have:
        \begin{align*}
            &\det\left(V_{1, \dots, 1}(\Lambda_1,\Lambda_{m_1 + 1},\dots,\Lambda_n)_{j,m_1}\right)\\
            &= \det(V_{1, \dots, 1}(\Lambda_{m_1 + 1},\dots,\Lambda_n)) 
                \cdot 
                \det
                    \left(
                    \left[
                        Q_{q - 1, p + 1 - q, 0}(\Lambda_{2},\dots, \Lambda_q)
                    \right]^{q = 2,\dots,n}_{p = 1,\dots,n, p \neq j}
                    \right)\\
            &= \left(\prod_{2 \leq \beta < \alpha \leq n} \left(\Lambda_{\alpha} - \Lambda_{\beta}\right)\right)
            \cdot 
                \det
                    \left(
                        \left[
                            Q_{q - 1, p + 1 - q, 0}(\Lambda_{2},\dots, \Lambda_q)
                        \right]^{q = 2,\dots,n}_{p = 1,\dots,n, p \neq j}
                    \right).
        \end{align*}
        Moreover, the function 
        \[
            \sum_{n \in \mathbb{N}} \C^n \to \sum_{n \in \mathbb{N}} \C^{n \times (n - 1)},
            \;
            (\lambda_1,\dots,\lambda_n) 
                \mapsto 
                \left[
                    Q_{q - 1, p + 1 - q, 0}(\lambda_{2},\dots, \lambda_q)
                \right]^{q = 2,\dots,n}_{p = 1,\dots,n}
        \]
        is computable in polynomial time.
        \item Assume that $m_1 \geq 2$.
        Then 
        \begin{align*}
            \det&\left(V_{m_1, 1, \dots, 1}(\Lambda_1,\dots,\Lambda_n)_{j,m_1}\right)=\\
             \det&\left(V_{{m_1 - 1}, 1, \dots, 1}(\Lambda_1,\dots,\Lambda_n)\right) \cdot \\
            &\det
                \left(
                    \left[
                        \left[Q_{1, p - q, q - 1}(\Lambda_1)\right]_{p = 1, \dots, n, p \neq j}^{q = 1,\dots, m_1 - 1}; 
                        \left[Q_{q - m_1 + 1, n - q + 1, m_1 - 2}(\Lambda_1,\dots,\Lambda_{q + 2}) \right]_{p = 1, \dots, n, p \neq j}^{q = m_1 + 1,\dots, n}
                    \right]
                \right)=\\
            &\left(\prod_{1 < \gamma \leq n}\left(\Lambda_\gamma - \Lambda_1\right)^{m_1 - 1}\right)
            \left(\prod_{1 < \beta < \alpha \leq n} \left(\Lambda_{\alpha} - \Lambda_{\beta}\right)\right)\cdot \\
            &\det
                \left(
                \left[
                    \left[Q_{1, p - q,q - 1}(\Lambda_1)\right]_{p = 1, \dots, n, p \neq j}^{q = 1,\dots, m_1 - 1}; 
                    \left[Q_{q - m_1 + 1, n - q + 1, m_1 - 2}(\Lambda_1,\dots,\Lambda_{q + 2}) \right]_{p = 1, \dots, n, p \neq j}^{q = m_1 + 1,\dots, n}
                \right]
                \right).
        \end{align*}
        Moreover, the function 
        \begin{align*}
            &\sum_{n \in \mathbb{N}} \C^n \to \sum_{n \in \mathbb{N}} \C^{n \times (n - 1)},
            \\
            &(\lambda_1,\dots,\lambda_n) \mapsto 
            \left[
                    \left[Q_{1, p - q,q - 1}(\lambda_1)\right]_{p = 1, \dots, n}^{q = 1,\dots, m_1 - 1}; 
                    \left[Q_{q - m_1 + 1, n - q + 1, m_1 - 2}(\lambda_1,\dots,\lambda_{q + 2}) \right]_{p = 1, \dots, n}^{q = m_1 + 1,\dots, n}
            \right]
        \end{align*}
        is computable in polynomial time.
    \end{enumerate}
\end{Lemma}
\begin{proof}
    We only show the second claim. The proof of the first claim is essentially identical.

    Let $W_{1, 0}$ be the matrix which is obtained from $\widetilde{V}_{m_1,1,\dots,1}(\Lambda_1,\dots,\Lambda_n)$
    by deleting the $m_1^{\text{st}}$ column.
    
    Observe that the mappings
    \[
        \mathcal{M}_j \colon K^{n \times (n - 1)} \to K, M \mapsto \det(M_{j}),
    \]
    where $M_j$ is obtained from $M$ by deleting the $j^{\text{th}}$ row, are alternating linear forms.
    Further, we have 
    \[
    \det\left(
        \widetilde{V}_{m_1 -1, 1,\dots,1} (\Lambda_1,\Lambda_{m_1 + 1}\dots,\Lambda_n)
    \right)
    =
    \mathcal{M}_{n}(W_{1, 0}).
    \]

    Thus, consider an arbitrary multilinear form
    \[
        \mathcal{M} \colon K^{n \times (n - 1)} \to K.
    \]
    We then have:
    \[
        W_{1, 0} = 
        \left[
        \left[
            C^{n}_{1, n - q, q - 1}(\Lambda_1)
        \right]^{q = 1,\dots,m_1 - 1};
        \left[
            C^{n}_{1, n - 1, 0}(\Lambda_{j})
        \right]^{j = m_1 + 1,\dots,n}
        \right].
    \]
    Subtract the first column of $W_{1, 0}$ from columns $m_1, \dots, n$ and use Lemma \ref{Lemma: recursive equations for Q}.1 to obtain:
    \[
        \mathcal{M}(W_{1, 0}) = 
        \left(
            \prod_{j \geq m_1 + 1}(\Lambda_j - \Lambda_1)
        \right)
        \cdot 
        \mathcal{M}(W_{1, 1}),
    \]
    where 
    \[
        W_{1, 1} = 
            \left[
                \left[
                    C^{n}_{1, n - q, q - 1}(\Lambda_1)
                \right]^{q = 1,\dots,m_1 - 1};
                \left[
                    C^{n}_{2, n - 2, 0}(\Lambda_1, \Lambda_{j})
                \right]^{j = m_1 + 1,\dots,n}
            \right].
    \]
    Subtract the $m_1^{\text{st}}$ column from the following columns to obtain with Lemma \ref{Lemma: recursive equations for Q}.1:
    \[
        \mathcal{M}(W_{1, 1}) = 
        \left(
            \prod_{j \geq m_1 + 2}(\Lambda_j - \Lambda_{m_1 + 1})
        \right)
        \cdot 
        \mathcal{M}(W_{1, 2}),
    \]
    where 
    \begin{align*}
        &W_{1, 2} =\\ 
            &\left[
                \left[
                    C^{n}_{1, n - q, q - 1}(\Lambda_1)
                \right]^{q = 1,\dots,m_1 - 1};
                \left[
                    C^n_{2, n - 2, 0}(\Lambda_1, \Lambda_{m_1 + 1})
                \right];
                \left[
                    C^n_{3, n - 3, 0}\left(\Lambda_1, \Lambda_{m_1 + 1}, \Lambda_j\right)
                \right]^{j = m_1 + 2,\dots,n}
            \right].
    \end{align*}
    Proceeding in this manner, we eventually obtain:
    \[
        \mathcal{M}(W_{1, 0}) = 
        \left(
            \prod_{1 \leq \beta < \alpha \leq n}(\Lambda_\alpha - \Lambda_\beta)
        \right)
        \cdot 
        \mathcal{M}(W_{1, n - m_1}),
    \]
    where 
    \begin{align*}
        &W_{1, n - m_1} =\\ 
            &\left[
                \left[
                    C^{n}_{1, n - q, q - 1}(\Lambda_1)
                \right]^{q = 1,\dots,m_1 - 1};
                \left[
                    C^{n}_{j - m_1 + 1, n - j + m_1 - 1, 0}(\Lambda_1, \Lambda_{m_1 + 1},\dots, \Lambda_j)
                \right]^{j = m_1 + 1,\dots,n}
            \right].
    \end{align*}
    If $m_1 = 2$, then the matrix $W_{1, n - m_1}$ is already of the required form.
    
    If $m_1 > 2$, subtract the second column of $W_{1, n - m_1}$ from the $m_1^{\text{st}}$ column and apply Lemma \ref{Lemma: recursive equations for Q}.2 to obtain:
    \[
        \mathcal{M}(W_{1, n - m_1}) = 
        \left(
                \Lambda_{m_1 + 1} - \Lambda_1
        \right)
        \cdot 
        \mathcal{M}(W_{2, 1}),
    \]
    where
    \begin{align*}
        &W_{2, 1} =\\ 
            &\Big[
                \left[
                    C^{n}_{1, n - q, q - 1}(\Lambda_1)
                \right]^{q = 1,\dots,m_1 - 1};
                \left[
                    C^{n}_{2, n - 3, 1}(\Lambda_1, \Lambda_{m_1 + 1})
                \right];\\
             &{\hspace{1cm}}\left[
                    C^{n}_{j - m_1 + 1, n - j + m_1 - 1, 0}(\Lambda_1, \Lambda_{m_1 + 1},\dots, \Lambda_j)
                \right]^{j = m_1 + 2,\dots,n}
            \Big].
    \end{align*}
    Now, subtract the $m_1^{\text{st}}$ column from the next column to obtain:
    \[
        \mathcal{M}(W_{1, n - m_1}) = 
        \left(
                \Lambda_{m_1 + 2} - \Lambda_1
        \right)
        \cdot 
        \mathcal{M}(W_{2, 2}),
    \]
    where
    \begin{align*}
        &W_{2, 2} =\\ 
            &\Big[
                \left[
                    C^{n}_{1, n - q, q - 1}(\Lambda_1)
                \right]^{q = 1,\dots,m_1 - 1};
                \left[
                    C^{n}_{2, n - 3, 1}(\Lambda_1, \Lambda_{m_1 + 1})
                \right];\\
                &{\hspace{1cm}}\left[
                    C^{n}_{3, n - 4, 1}(\Lambda_1, \Lambda_{m_1 + 1}, \Lambda_{m_1 + 2})
                \right];
                \left[
                    C^{n}_{j - m_1 + 1, n - j + m_1 - 1, 0}(\Lambda_1, \Lambda_{m_1 + 1},\dots, \Lambda_j)
                \right]^{j = m_1 + 3,\dots,n}
            \Big].
    \end{align*}
    Proceed in this manner to obtain:
    \[
        \mathcal{M}(W_{1, n - m_1}) = 
        \left(
                \prod_{j \geq m_1 + 1}\left(\Lambda_{j} - \Lambda_1\right)
        \right)
        \cdot 
        \mathcal{M}(W_{2, n - m_1})
    \]
    where 
    \begin{align*}
        &W_{2, n - m_1} =\\ 
            &\left[
                \left[
                    C^{n}_{1, n - q, q - 1}(\Lambda_1)
                \right]^{q = 1,\dots,m_1 - 1};
                \left[
                    C^{n}_{j - m_1 + 1, n - j + m_1 - 2, 1}(\Lambda_1, \Lambda_{m_1 + 1},\dots, \Lambda_j)
                \right]^{j = m_1 + 1,\dots,n}
            \right].
    \end{align*}

    If $m_1 > 3$, we proceed in the same manner, starting with the third column of $W_{2, n - m_1}$ to eventually obtain: 
    \[
        \mathcal{M}(W_{1, 0}) = 
        \left(
            \prod_{1 \leq \beta < \alpha \leq n}(\Lambda_\alpha - \Lambda_\beta)
        \right)
        \cdot 
        \left(
            \prod_{j \geq m_1 + 1}
                \left( 
                    \Lambda_j - \Lambda_1
                \right)^{m_1 - 2}
        \right)
        \mathcal{M}\left(W_{m_1 - 1, n - m_1}\right)
    \]
    where 
    \begin{align*}
        &W_{m_1 - 1, n - m_1} =\\ 
            &\left[
                \left[
                    C^{n}_{1, n - q, q - 1}(\Lambda_1)
                \right]^{q = 1,\dots,m_1 - 1};
                \left[
                    C^{n}_{j - m_1 + 1, n - j + 1, m_1 - 2}(\Lambda_1, \Lambda_{m_1 + 1},\dots, \Lambda_j)
                \right]^{j = m_1 + 1,\dots,n}
            \right].
    \end{align*}
    This finishes the proof of the claimed identity.
    The claim about 
    $\det\left(
        \widetilde{V}_{m_1 -1, 1,\dots,1} (\Lambda_1,\Lambda_{m_1 + 1}\dots,\Lambda_n)
    \right)$
    follows by noticing that the matrix which is obtained from
    $W_{m_1 - 1, n - m_1}$ by removing the last row is a lower echelon matrix whose diagonal entries are all equal to $1$.

    It remains to show that the entries of the matrix 
    \begin{align*}
        &W_{m_1 - 1, n - m_1}(\lambda_1,\lambda_{m_1 + 1},\dots,\lambda_n)\\
        &=
        \left[
            \left[
                C^{n}_{1, n - q, q - 1}(\lambda_1)
            \right]^{q = 1,\dots,m_1 - 1};
            \left[
                C^{n}_{j - m_1 + 1, n - j + 1, m_1 - 2}(\lambda_1, \lambda_{m_1 + 1},\dots, \lambda_j)
            \right]^{j = m_1 + 1,\dots,n}
        \right]
    \end{align*}
    are uniformly polytime computable in $\lambda_1,\lambda_{m_1 + 1},\dots,\lambda_n$.

    To this end, observe that the polynomials
    $Q_{n, m, i}$ 
    satisfy the recursive equation
    \[
        Q_{n + 1, m, i}(X_1,\dots,X_{n + 1}) 
        =
        \sum_{k = 0}^m X_{n + 1}^k Q_{n, m - k, i}(X_1,\dots,X_n).
    \]
    We can hence compute a tableau that contains all entries of
        $W_{m_1 - 1, n - m_1}(\lambda_1,\lambda_{m_1 + 1},\dots,\lambda_n)$
    as follows:
    
    \begin{enumerate}
        \item 
    Start with the numbers 
    $Q_{1, m, i}(\lambda_1) = \binom{m + i}{i} \lambda_1^m$
    for 
    $m = 0,\dots,n - 1$ and $i = 0,\dots,m_1 - 2$.
        \item
    Having computed the numbers  
    $Q_{\ell + 1, k, i}\left(\lambda_1, \lambda_{m_1 + 1},\dots,\lambda_{m_1 + \ell}\right)$
    for $\ell \geq 0$, compute
    \[
        Q_{\ell + 2, m, i}\left(\lambda_1, \lambda_{m_1 + 1}, \dots, \lambda_{m_1 + \ell}, \lambda_{m_1 + \ell + 1}\right) 
        =
        \sum_{k = 0}^m \lambda_{m_1 + \ell + 1}^k Q_{\ell + 1, m - k, i}(\lambda_1, \lambda_{m_1 + 1}, \dots,\lambda_{m_1 + \ell}).
    \]
    \end{enumerate}
    It is easy to prove that the size of the numbers 
    $Q_{\ell + 2, m, i}\left(\lambda_1, \lambda_{m_1 + 1}, \dots, \lambda_{m_1 + \ell}, \lambda_{m_1 + \ell + 1}\right)$ 
    grows polynomially in $\ell$, $m$, and $i$.
    From this, it is easy to deduce that the full tableau with 
    $\ell = 0, \dots, n - 1$, $m = 0,\dots,n - 1$, $i = 0,\dots, m_1 - 2$
    is computable in polynomial time.
    The matrix entries can then be easily computed from this tableau.
\end{proof}

%% file: Coefficient-Function-Theorem.tex
\section{Proof of Theorem \ref{Theorem: coefficient functions}}\label{Appendix: coefficient function theorem}

\begin{enumerate}
    \item 
            Assume first that $m_1 = 1$.
            Let
            \[ 
                W_{n, 1, j} =
                \left[
                    Q_{q - 1, p + 1 - q, 0}(\Lambda_{2},\dots, \Lambda_q)
                \right]^{q = 2,\dots,n}_{p = 1,\dots,n, p \neq j}
            \]
            By Lemma \ref{Lemma: minors of generalised Vandermonde matrix} we have:
            \begin{align*}
                F_{1,\dots,1}
                &= 
                \frac{\sum_{j = 1}^n (-1)^{j + 1} \det\left(\widetilde{V}_{1,\dots,1}(\Lambda_1, \dots,\Lambda_n)_{j, 1}\right) \cdot U_{j - 1} }
                    {\det \widetilde{V}_{1,\dots,1}(\Lambda_1, \dots,\Lambda_n)}\\
                &=
                \frac{\sum_{j = 1}^n (-1)^{j + 1} 
                    \det\left(V_{1,\dots,1}(\Lambda_1, \dots,\Lambda_n)_{j, 1}\right) \cdot U_{j - 1} }
                {\det V_{1,\dots,1}(\Lambda_1, \dots,\Lambda_n)}\\
                &=
                \frac{\sum_{j = 1}^n (-1)^{j + 1} \det\left(V_{1, \dots, 1}(\Lambda_2,\dots,\Lambda_n)\right)
                \cdot 
                \det \left(W_{n, 1, j}\right) \cdot U_{j + 1} }
                    {\det V_{1,\dots,1}(\Lambda_1, \dots,\Lambda_n)}.
            \end{align*}
            The matrix $V_{1, \dots, 1}(\Lambda_2,\dots,\Lambda_n)$ has determinant 
            \[
                \prod_{2 \leq j < k \leq n} \left( \Lambda_k - \Lambda_j \right).
            \]
            The matrix $V_{1, \dots, 1}(\Lambda_1,\dots,\Lambda_n)$ has determinant 
            \[
                \prod_{1 \leq j < k \leq n} \left( \Lambda_k - \Lambda_j \right).
            \]
            The matrix $W_{n, 1, n}$ is a lower echelon matrix whose diagonal entries are equal to $1$.
            It follows that $\det\left(W_{n, 1, n}\right) = 1$.
            We hence obtain 
            \[
                F_{1,\dots,1} = 
                \prod_{k = 2}^n \left( \Lambda_k - \Lambda_1 \right)^{-1}
                \left(
                    (-1)^{n + 1} U_{n - 1}
                    +
                    \sum_{j = 1}^{n - 1} (-1)^{j + 1} \det \left(W_{n, 1, j}\right) U_{j + 1} 
                \right).
            \]
            Letting $A_{n,1,j}(\Lambda_1,\dots,\Lambda_n) = \det \left(W_{n, 1, j}\right)$
            shows that $F_{1,\dots,1}$ has the desired shape.

            Assume now that $m_1 \geq 2$.
            Let 
            \[
                W_{n, m_1, j} = 
                    \left[
                        \left[Q_{1, p - q,q - 1}(\Lambda_1)\right]_{p = 1, \dots, n, p \neq j}^{q = 1,\dots, m_1 - 1}; 
                        \left[Q_{q - m_1 + 1, n - q + 1, m_1 - 2}(\Lambda_1,\dots,\Lambda_{q + 2}) \right]_{p = 1, \dots, n, p \neq j}^{q = m_1 + 1,\dots, n}
                    \right]
            \]
            By Lemma \ref{Lemma: minors of generalised Vandermonde matrix} we have:
            \begin{align*}
                F_{m_1, 1,\dots,1}
                &= 
                \frac{\sum_{j = 1}^n (-1)^{j + m_1} \det\left(\widetilde{V}_{m_1,1,\dots,1}(\Lambda_1, \Lambda_{m_1 + 1}, \dots,\Lambda_n)_{j, m_1}\right) \cdot U_{j - 1} }
                    {\det \widetilde{V}_{m_1,1,\dots,1}(\Lambda_1, \Lambda_{m_1 + 1}, \dots,\Lambda_n)}\\
                &=
                \frac{\sum_{j = 1}^n (-1)^{j + m_1} 
                    \left(\prod_{k = 0}^{m_1 - 2}k!\right)\det\left(V_{m_1,1,\dots,1}(\Lambda_1, \Lambda_{m_1 + 1}, \dots,\Lambda_n)_{j, m_1}\right) \cdot U_{j - 1} }
                {\left(\prod_{k = 0}^{m_1 - 1}k!\right) \det V_{m_1,1,\dots,1}(\Lambda_1, \Lambda_{m_1 + 1}, \dots,\Lambda_n)}\\
                &=
                \frac{\sum_{j = 1}^n (-1)^{j + m_1} \det\left(V_{m_1 - 1, 1, \dots, 1}(\Lambda_1,\Lambda_{m_1 + 1},\dots,\Lambda_n)\right)
                \cdot 
                \det \left(W_{n, m_1, j}\right) \cdot U_{j + m_1} }
                    {(m_1 - 1)! \det V_{m_1,1,\dots,1}(\Lambda_1, \Lambda_{m_1 + 1}, \dots,\Lambda_n)}.
            \end{align*}
            We have 
            \[ 
                \det V_{m_1,1,\dots,1}(\Lambda_1, \Lambda_{m_1 + 1}, \dots,\Lambda_n) = 
                    \left(\prod_{j = m_1 + 1}^n (\Lambda_j - \Lambda_1)^{m_1}\right)
                    \left(\prod_{m_1 + 1 \leq j < k \leq n} (\Lambda_k - \Lambda_j)\right)
            \]
            and 
            \[ 
                \det V_{m_1 - 1,1,\dots,1}(\Lambda_1, \Lambda_{m_1 + 1}, \dots,\Lambda_n) = 
                    \left(\prod_{j = m_1 + 1}^n (\Lambda_j - \Lambda_1)^{m_1 - 1}\right)
                    \left(\prod_{m_1 + 1 \leq j < k \leq n} (\Lambda_k - \Lambda_j)\right).
            \] 
            The matrix $W_{n, m_1, n}$ is a lower echelon matrix whose diagonal entries are all equal to $1$.
            It follows that $\det(W_{n, m_1, n}) = 1$.
            Letting $A_{n, m_1, j} = \det(W_{n, m_1, j})$ for $j = 1,\dots, n - 1$ shows that $F_{m_1,1,\dots,1}$ has the desired shape.
            
            Using that, according to Lemma \ref{Lemma: minors of generalised Vandermonde matrix}, 
            $W_{j, m_1, n}(\lambda_1,\lambda_{m+1},\dots,\lambda_n)$ 
            is uniformly polynomial-time computable in $n$ and $(\lambda_1,\lambda_{m+1},\dots,\lambda_n)$
            we obtain that $A_{n, m_1, j}(\lambda_1,\lambda_{m+1},\dots,\lambda_n)$
            is uniformly polynomial-time computable in $n$ and $(\lambda_1,\lambda_{m+1},\dots,\lambda_n)$.
            It suffices to observe that the determinant of an $n \times n$ matrix is uniformly polynomial-time computable in $n$ and the matrix entries.
    \item By the previous item, the function $G_{m_1,n}$ is computable in polynomial time, uniformly in $n$ and the bitsize of the input vector. 
    In other words, there exist a polynomial $\omega \in \N[X,Y,Z]$
    and an algorithm which takes as input $n \geq 1$, $p \in \N$, and a complex vector $x \in \C^{n - m_1 + 1} \times \C^n$ and outputs
    a rational approximation of 
    $G_{m_1,n}(x)$
    to error $2^{-p-1}$
    within $\omega(n, C, p)$ steps,
    where $C$ is a bound on the logarithm of the absolute values of the entries of $x$.
    In particular, this algorithm can query its input vector to accuracy at most $2^{-\omega(n, C, p + 1)}$
    before producing an approximation to error $2^{-p-1}$.
    If two input vectors have distance less than $2^{-\omega(n, C, p)}$, then the algorithm outputs the same rational approximation 
    to error $2^{-p - 1}$ on both inputs. 
    The claim follows if we put $\Omega(X,Y,Z) = \omega(X,Y, Z + 1)$.
    \item 
    It remains to prove the computability of the function $\gamma$.

    Suppose we are given complex numbers 
    $(\lambda_1,\lambda_{m_1 + 1},\dots,\lambda_n)$
    and 
    $(u_0,\dots,u_{n - 1})$
    such that $\lambda_j \neq \lambda_1$ for $j \geq m_1 + 1$.
    Using Vieta's formulas we can compute the coefficients of the polynomial 
    \[
        \chi(z)
        =
        (z - \lambda_1)^{m_1} \prod_{j = m_1 + 1}^n (z - \lambda_n)
        =
        z^n + c_{n - 1} z^{n - 1} + \dots + c_1 z + c_0
        .
    \]
    Consider the matrix 
    \[
        A = 
        \begin{bmatrix}
            0         & 1          &            &            \\ 
            \vdots    & \ddots     & \ddots     &   \\ 
            0         & \dots      & 0          & 1      \\ 
            -c_0      & \dots      & -c_{n - 2} & -c_{n - 1} \\ 
        \end{bmatrix}
    \]
    and the associated initial value problem
    \begin{align}\label{eq: Vandermonde n-d initial value problem}
        &\dot{x}(t) = A x(t) \\
        &x(0) = u.
    \end{align}
    Let $x$ be a solution to \eqref{eq: Vandermonde n-d initial value problem}.
    Then the function $x_1(t)$ is a solution to the initial value problem \eqref{eq: differential equation}
    (with complex coefficients).

    By construction, the number $\lambda_1$ is an eigenvalue of $A$ with algebraic multiplicity $m_1$.
    By observing that $\chi(z)$ and the first $m_1 - 1$ derivatives of $\chi(z)$ vanish in $\lambda_1$,
    we find that the geometric multiplicity of $\lambda_1$ is equal to $1$
    and that a Jordan chain for $\lambda_1$ is given by the generalised Vandermonde block 
    $[\lambda_1]_{n, m_1}$ of size $n \times m_1$ in $\lambda_1$.
    It follows from \cite[Lemma 4.2, Lemma 4.3]{DPLSR20} that we can compute matrices $S$ and $S^{-1}$ such that 
    \[
        S A S^{-1} = 
        \begin{bmatrix}
            J & 0 \\
            0 & B    
        \end{bmatrix},
    \]
    where $J$ is a Jordan block for $\lambda_1$ of size $m_1$
    and all generalised eigenvectors of $S A S^{-1}$ for the eigenvalues $\lambda_{m_1 + 1},\dots, \lambda_n$
    lie in the span of the standard unit vectors $e_{m_1 + 1},\dots, e_n$.
    Note that the results in \cite{DPLSR20} are stated only for real matrices, but their proof equally applies to complex matrices.
    
    Consider the initial value problem 
    \begin{align*}
        &\dot{y}(t) = S A S^{-1} y(t) \\
        &y(0) = S u. 
    \end{align*}
    The unique solution $y$ of this initial value problem satisfies 
    $y(t) = S x(t)$,
    where $x$ is the unique solution of \eqref{eq: Vandermonde n-d initial value problem}.

        The simple shape of $S A S^{-1}$ allows us to solve the initial value problem 
        for the first $m_1$ components explicitly. 
        We have 
        \[
            y_j(t) = \sum_{k = 0}^{m_1 -j} \frac{u_{m_1 - k - 1}}{(m_1 - j - k)!} t^{m_1 - j - k} e^{\lambda_1 t}
        \]
        for $j = 1,\dots,m_1$.

    Let $s_{j,k}$, $j = 1,\dots,n$, $k = 1,\dots, n$ denote the entries of the matrix $S^{-1}$.
    We have 
    \[
        x_1(t) = s_{1,1} y_1(t) + s_{1,2} y_2(t) + \dots + s_{1,n} y_n(t),
    \]
    where $\sum_{k = m_1 + 1}^{n} s_{1,k} y_{k}(t)$ is a linear combination of terms of the form 
    $t^\ell e^{\lambda_j t}$ with $j \geq m_1 + 1$.
    It follows that the coefficient of $t^{m_1 - 1} e^{\lambda_1 t}$ in the unique solution of 
    \eqref{eq: differential equation}
    is equal to
    \[
        \frac{s_{1,1} u_{m_1 - 1}}{(m_1 - 1)!}.
    \]
    Since the matrix $S$ is uniformly computable in the input data, so is this number.
    This finishes the proof.
\end{enumerate}

%% file: Proposition-relating-coefficient-functions.tex
\section{Proof of Proposition \ref{Proposition: coefficient functions for equal Cauchy problems}}
\label{Appendix: Proposition relating coefficient functions}

By continuity, it suffices to prove the equality when the numbers $\lambda_1,\dots,\lambda_{n_1}$ are pairwise distinct.
    Then 
    $F_{m, n_1}(\lambda_1, \lambda_{m + 1},\dots,\lambda_{n_1}, u_0,\dots,u_{n_1 - 1})$
    and 
    $F_{m, n_2}(\lambda_1, \lambda_{m + 1},\dots,\lambda_{n_2}, u_0,\dots,u_{n_2 - 1})$
    are well-defined with 
    \begin{align*}
        F_{m, n_1}&(\lambda_1, \lambda_{m + 1},\dots,\lambda_{n_1}, u_0,\dots,u_{n_1 - 1}) = \\
            &\frac{1}{(m - 1)!}
            \left(
                \prod_{j = m + 1}^{n_1}\left( \lambda_j - \lambda_1 \right)
            \right)^{-1}
            G_{m, n_1}(\lambda_1, \lambda_{m + 1},\dots,\lambda_{n_1}, u_0,\dots,u_{n_1 - 1})
    \end{align*}
    and 
    \begin{align*}
        F_{m, n_2}&(\lambda_1, \lambda_{m + 1},\dots,\lambda_{n_2}, u_0,\dots,u_{n_2 - 1}) = \\
            &\frac{1}{(m - 1)!}
            \left(
                \prod_{j = m + 1}^{n_2}\left( \lambda_j - \lambda_1 \right)
            \right)^{-1}
            G_{m, n_2}(\lambda_1, \lambda_{m + 1},\dots,\lambda_{n_2}, u_0,\dots,u_{n_2 - 1}).
    \end{align*}
    It hence suffices to prove that the numbers
    \[F_{m, n_1}(\lambda_1, \lambda_{m + 1},\dots,\lambda_{n_1}, u_0,\dots,u_{n_1 - 1})\]
    and 
    \[F_{m, n_2}(\lambda_1, \lambda_{m + 1},\dots,\lambda_{n_2}, u_0,\dots,u_{n_2 - 1})\]
    are equal.

    Consider the Cauchy problem with characteristic polynomial $\prod_{j = 1}^{n_1} (z - \lambda_j)$
    and initial values $u_0,\dots,u_{n_1 - 1}$.
    By assumption, the solution to this Cauchy problem already satisfies the linear differential equation 
    with characteristic polynomial $\prod_{j = 1}^{n_2} (z - \lambda_j)$.
    It follows that the solution to this Cauchy problem must be equal to the solution of the Cauchy problem 
    with characteristic polynomial $\prod_{j = 1}^{n_2} (z - \lambda_j)$
    and initial values $u_0,\dots,u_{n_2 - 1}$.
    But this implies that
    \[F_{m, n_1}(\lambda_1, \lambda_{m + 1},\dots,\lambda_{n_1}, u_0,\dots,u_{n_1 - 1})\]
    and 
    \[F_{m, n_2}(\lambda_1, \lambda_{m + 1},\dots,\lambda_{n_2}, u_0,\dots,u_{n_2 - 1})\]
    must be equal
    by definition of the functions $F_{\ell, n}$.